\documentclass{article}

\usepackage[affil-it]{authblk}
\usepackage[usenames,dvipsnames]{xcolor}
\usepackage{amsfonts}
\usepackage{amsmath,amsthm,amssymb,dsfont}
\usepackage{enumerate}
\usepackage{graphicx}	
\usepackage{subcaption}
\usepackage[margin=3cm]{geometry}
\usepackage{url}
\usepackage{todonotes}
\usepackage{bbm}

\usepackage{tikz}

\usepackage{pifont}
\usepackage{multirow}
\usepackage{makecell}

\usepackage{epsfig}
\usetikzlibrary{shapes.symbols,patterns} 
\usepackage{pgfplots}
\pgfplotsset{compat=1.10}
\usepgfplotslibrary{fillbetween}

\definecolor{linkblue}{HTML}{001487}
\usepackage{hyperref}
\hypersetup{colorlinks=true,citecolor=linkblue,linkcolor=linkblue,filecolor=linkblue,urlcolor=linkblue,breaklinks=true}

\usepackage{nicefrac}
\usepackage{mathtools}
\usepackage[nameinlink,capitalize,noabbrev]{cleveref}

\usepackage{algorithm}
\usepackage{algorithmic}

\usepackage{mdframed}
\usepackage{aligned-overset}
\usepackage{circuitikz}

\theoremstyle{plain}
\newtheorem{theorem}{Theorem}[section]
\newtheorem{lemma}{Lemma}[section]

\newtheorem{fact}[lemma]{Fact}

\theoremstyle{definition}

\newtheorem{remark}[lemma]{Remark}

\usepackage{tcolorbox}
\tcolorboxenvironment{definition}{}
\tcolorboxenvironment{theorem}{colback=pink!25!white,colframe=pink!100!black}
\tcolorboxenvironment{lemma}{colback=yellow!5!white,colframe=yellow!75!black}
\tcolorboxenvironment{corollary}{colback=Dandelion!5!white,colframe=Dandelion!75!black}
\tcolorboxenvironment{proposition}{colback=Emerald!5!white,colframe=Emerald!75!black}
\tcolorboxenvironment{claim}{colback=RoyalBlue!5!white,colframe=RoyalBlue!75!black}
\tcolorboxenvironment{conjecture}{colback=red!5!white,colframe=red!75!black}
\tcolorboxenvironment{example}{colback=white,colframe=lightgray}
\tcolorboxenvironment{remark}{colback=white,colframe=lightgray}
\tcolorboxenvironment{question}{colback=lightgray!5!white,colframe=lightgray!75!black}

\DeclareRobustCommand{\abbrevcrefs}{%
\Crefname{theorem}{Thm.}{Thms.}%
\Crefname{corollary}{Cor.}{Cors.}%
\Crefname{lemma}{Lem.}{Lems.}%
\Crefname{proposition}{Prop.}{Props.}%
\Crefname{equation}{Eq.}{Eqs.}%
\Crefname{example}{Ex.}{Exs.}%
\Crefname{remark}{Rmk.}{Rmks.}%
\Crefname{fact}{Fact}{Facts}%
}

\DeclareMathOperator*{\argmin}{argmin}

\DeclareRobustCommand{\Cshref}[1]{{\abbrevcrefs\Cref{#1}}}

\newcommand*{\ee}{\mathrm{e}}

\newcommand*{\cA}{\mathcal{A}}
\newcommand*{\cB}{\mathcal{B}}
\newcommand*{\cC}{\mathcal{C}}

\newcommand*{\cE}{\mathcal{E}}

\newcommand*{\cH}{\mathcal{H}}

\newcommand*{\cM}{\mathcal{M}}
\newcommand*{\cP}{\mathcal{P}}

\newcommand*{\cR}{\mathcal{R}}
\newcommand*{\cS}{\mathcal{S}}

\newcommand*{\cX}{\mathcal{X}}

\newcommand*{\N}{\mathbb{N}}

\newcommand*{\R}{\mathbb{R}}

\newcommand*{\St}{\mathrm{S}}

\newcommand*{\eps}{\varepsilon}
\newcommand*{\diag}{\mathrm{diag}}

\newcommand*{\id}{\mathds{1}}
\newcommand*{\poly}{\mathrm{poly}}

\newcommand*{\spec}{\mathrm{spec}}

\newcommand*{\tr}{\mathrm{tr}}
\newcommand*{\ket}[1]{| #1 \rangle}
\newcommand*{\bra}[1]{\langle #1 |}

\newcommand{\proj}[1]{|#1\rangle\!\langle #1|}

\newcommand*{\TPCP}{\mathrm{TPCP}}

\newcommand*{\MD}{D_{\mathbb{M}}}

\newcommand*{\ci}{\mathrm{i}} 
\newcommand*{\di}{\mathrm{d}} 

\newcommand{\norm}[1]{\left\lVert#1\right\rVert}

\DeclareMathOperator*{\argmax}{arg\,max}

\usepackage{float}

 \allowdisplaybreaks

\title{Uhlmann's theorem for relative entropies}

 \author{\normalsize Giulia Mazzola$^{1}$, David Sutter$^{2}$, and Renato Renner$^{1}$}
  \affil{\small $^{1}$Institute for Theoretical Physics, ETH Zurich\\
  $^{2}$IBM Quantum, IBM Research Europe -- Zurich
 }
 \date{}

\newcommand{\mylabel}[2]{#2\def\@currentlabel{#2}\label{#1}}

\begin{document}

\maketitle

\begin{abstract}
Uhlmann's theorem states that, for any two quantum states $\rho_{AB}$ and $\sigma_A$, there exists an extension $\sigma_{AB}$ of $\sigma_A$ such that the fidelity between $\rho_{AB}$ and $\sigma_{AB}$ equals the fidelity between their reduced states $\rho_A$ and $\sigma_A$. In this work, we generalize Uhlmann's theorem to $\alpha$-R\'enyi relative entropies for $\alpha \in [\frac{1}{2},\infty]$, a family of divergences that encompasses fidelity, relative entropy, and max-relative entropy corresponding to $\alpha=\frac{1}{2}$, $\alpha=1$, and $\alpha=\infty$, respectively.
\end{abstract}

\section{Introduction} \label{sec_intro}
Consider two density matrices $\rho_{AB}$ and $\sigma_A$ on finite-dimensional Hilbert spaces $A \otimes B$ and $A$, respectively. We want to find an extension $\sigma_{AB}$ of $\sigma_A$ that is closest to $\rho_{AB}$ according to some divergence $\mathbb{D}$ such as the relative entropy.\footnote{A divergence may be viewed as a distance in the sense that it is positive definite. However, it does not necessarily satisfy a triangle inequality, nor does it have to be symmetric.}
In mathematical terms, we are interested in the quantity 
\begin{align} \label{eq_def_Ua}
\mathbb{D}(\rho_{AB} \| \cC^{\sigma_A}_{AB} ):=\inf_{\sigma_{AB} \in \cC^{\sigma_A}_{AB}} \mathbb{D}(\rho_{AB}\| \sigma_{AB}) \, ,
\end{align}
where the infimum\footnote{All divergences considered in this manuscript are lower semi-continuous~\cite[Proposition~III.11]{MH24}. Since we assume finite-dimensional Hilbert spaces, the compactness of $\cC_{AB}^{\sigma_A}$ implies that the infimum on the right-hand side of~\cref{eq_def_Ua} is attained and may thus be replaced by a minimum.\label{footnote_minimum}} ranges over all extensions $\sigma_{AB}$ of $\sigma_A$, i.e.,
\begin{align} \label{eq_set_C}
\cC_{AB}^{\sigma_A}:=\{\textnormal{density matrices } \sigma'_{AB}: \tr_B[\sigma'_{AB}]=\sigma_A\} \, .
\end{align}
In particular, our objective is to understand how \smash{$\mathbb{D}(\rho_{AB} \| \cC^{\sigma_A}_{AB} )$} is related to $\mathbb{D}(\rho_{A}\| \sigma_{A})$.

Most divergences used in quantum information satisfy the \emph{data-processing inequality} (DPI), which states that they cannot increase when applying the same quantum channel, such as a partial trace, to both arguments. We thus have $\mathbb{D}(\rho_{A}\| \sigma_{A})  \leq \mathbb{D}(\rho_{AB}\| \sigma_{AB})$, and hence,
\begin{align} \label{eq_DPI}
\mathbb{D}(\rho_{A}\| \sigma_{A}) \overset{\textnormal{DPI}}{\leq} \mathbb{D}(\rho_{AB} \| \cC^{\sigma_A}_{AB} )  \, .
\end{align}
As we will see, whether this inequality holds with equality depends on the particular divergence $\mathbb{D}$.

A popular divergence is the \emph{min-relative entropy}, defined as $D_{\min}(\rho \| \sigma):= - \log F(\rho, \sigma)$, where 
 \smash{$ F(\rho, \sigma):=\norm{\sqrt{\rho} \sqrt{\sigma}}_1^2$} is the fidelity and \smash{$\norm{X}_1:=\tr[\sqrt{X X^\dagger}]$} denotes the trace norm. 
 For this quantity, the celebrated Uhlmann theorem~\cite{Uhl76} (see also~\cite[Corollary~3.14]{marco_book}) states that~\cref{eq_DPI} holds with equality.
\begin{theorem}[{Uhlmann's theorem~\cite{Uhl76}}] \label{thm_uhlmann}
Let $\rho_{AB}$ and $\sigma_A$ be two density matrices. Then, 
\begin{align} \label{eq_Uhlmann_our_Not}
D_{\min}(\rho_{A}\|\sigma_{A})  = D_{\min}(\rho_{AB}\|\cC_{AB}^{\sigma_A}) \, .
\end{align}
Furthermore, if $\rho_{AB}$ is pure and $\dim(A) \leq \dim(B)$, then the optimizer $\sigma_{AB} \in \cC_{AB}^{\sigma_A}$ of $D_{\min}(\rho_{AB}\|\cC_{AB}^{\sigma_A})$ can be chosen pure as well.
\end{theorem}

Uhlmann's theorem is widely used as it allows us to work on purifications, where the min-relative entropy simplifies. More precisely, suppose that we want to compute $D_{\min}(\rho_A\| \sigma_A)$ for two mixed states. Then we can choose any purification $\ket{\psi}_{AB}$ of $\rho_A$, and~\cref{eq_Uhlmann_our_Not} allows us to express $D_{\min}(\rho_A\| \sigma_A)$ in terms of $D_{\min}(\proj{\psi}\|\cC_{AB}^{\sigma_A})$. The latter quantity can be simpler to work with because the fidelity reduces to an inner product if one of the two arguments is pure, so that $D_{\min}(\rho_A\| \sigma_A) = -\log \bra{\psi}\sigma_{AB} \ket{\psi}$.

It was shown in~\cite[Proposition B.1]{GEAT_24} that~\cref{eq_Uhlmann_our_Not} remains true for another divergence, called the \emph{max-relative entropy}~\cite{renner_phd,datta09_b}, $D_{\max}(\rho \| \sigma):=\inf\{\lambda \in \R : \rho \leq 2^{\lambda} \sigma\}$, i.e., 
\begin{align} \label{eq_uhlmann_max}
D_{\max}(\rho_{A}\|\sigma_{A}) = D_{\max}(\rho_{AB}\| \cC_{AB}^{\sigma_A})  \, .
\end{align}
This raises the question if the Uhlmann-type identities given by~\cref{eq_Uhlmann_our_Not,eq_uhlmann_max} hold true for other relative entropies, such as the (sandwiched) \emph{R\'enyi relative entropy} $D_{\alpha}(\rho \| \sigma)$ for $\alpha \in [\frac{1}{2},1) \cup (1,\infty]$ (which we will define in~\cref{eq_sandwiched} below).
For $\alpha =\frac{1}{2}$ and $\alpha = \infty$ we recover the min- and max-relative entropy, respectively. For $\alpha \to 1$ we obtain the \emph{relative entropy} (defined in~\cref{eq_def_D}).
The R\'enyi relative entropy is monotonically increasing in $\alpha$~\cite[Theorem~7]{MLDSFT13}, which justifies the naming \smash{$D_{\min}=D_{\frac{1}{2}} \leq D_1 = D \leq D_{\infty}=D_{\max}$}. 

However, the most straightforward generalization of Uhlmann's theorem one may hope to achieve, namely that~\cref{eq_DPI} holds with equality for $D_\alpha$, cannot be valid for $\alpha \in (\frac{1}{2},\infty)$, as shown in~\cite[Appendix~B]{GEAT_24} together with the argument given in~\cref{app_Henrik}.
More precisely, there exist density matrices $\rho_{AB}$ and $\sigma_A$ such that~\cref{eq_DPI} is strict\footnote{\Cref{eq_DPI} can be strict for other divergences such as the Petz R\'enyi relative entropy~\cite{petz_Entropybook,marco_book} for $\alpha \in (0,2)$, as explained in~\cref{app_Henrik}.}, i.e.,
\begin{align} \label{eq_tony_counter_example}
D_{\alpha}(\rho_{A}\|\sigma_{A}) \overset{\textnormal{\cite[Appendix~B]{GEAT_24} and \Cref{app_Henrik}}}{<}    D_{\alpha}(\rho_{AB}\| \cC_{AB}^{\sigma_A} )  \quad \forall \alpha \in (\frac{1}{2},\infty) \, .
\end{align}
Nonetheless, as we will show in this work, we can still obtain Uhlmann-type equalities for these divergences, provided that we consider modifications obtained by regularization or measurement.

\paragraph{Main results.} 
\begin{enumerate}[(1)]
\item For any two density matrices $\rho_{AB}$, $\sigma_A$, and $\alpha \in [\frac{1}{2},\infty]$, we prove that\footnote{The assertion for $\alpha>1$ was already known from earlier work~\cite[Lemma~3.3]{GEAT_24}.}
\begin{align} \label{eq_main_res}
D_{\alpha}(\rho_{A}\|\sigma_{A})=   D_{\alpha}^{\infty}(\rho_{AB}\| \cC^{\sigma_A}_{AB} ) \, ,
\end{align}
for the regularized quantity\footnote{The limit in the regularization exists and can be replaced by an infimum due to Fekete's subadditivity lemma~\cite{fekete23}, since 
\begin{align*}D_{\alpha}\Big(\rho_{AB}^{\otimes n+m} \Big\| \cC^{\sigma_A^{\otimes n+m}}_{A_1^{n+m}B_1^{n+m}}\Big)\leq D_{\alpha}\Big(\rho_{AB}^{\otimes n}\Big\| \cC^{\sigma_A^{\otimes n}}_{A_1^nB_1^n}\Big)+D_{\alpha}\Big(\rho_{AB}^{\otimes m}\Big\| \cC^{\sigma_A^{\otimes m}}_{A_1^mB_1^m}\Big)\, . \end{align*}}
\begin{align} \label{eq_Da_regul}
D_{\alpha}^{\infty}(\rho_{AB}\| \cC^{\sigma_A}_{AB} )
:= \lim_{n \to \infty} \frac{1}{n} D_{\alpha}\big(\rho_{AB}^{\otimes n}\|\cC^{\sigma_A^{\otimes n}}_{A_1^nB_1^n}\big)
= \inf_{n \in \N} \frac{1}{n} D_{\alpha}\big(\rho_{AB}^{\otimes n}\|\cC^{\sigma_A^{\otimes n}}_{A_1^nB_1^n}\big) \, ,
\end{align}
where \smash{$\cC^{\sigma_A^{\otimes n}}_{A_1^nB_1^n}$} denotes the set of all extensions of $\sigma_A^{\otimes n}$ as defined in~\cref{eq_set_C}. We call~\cref{eq_main_res} a \textbf{regularized Uhlmann's theorem} for relative entropies and refer to~\cref{thm_tony_new_range,thm_uhlmann_rel_ent} for the rigorous statements. Furthermore, for $\alpha=1$, the optimizer in $D_{\alpha}^{\infty}(\rho_{AB}\| \cC^{\sigma_A}_{AB} )$ has an explicit form as shown in~\cref{rmk_explicit_optimizer_regul}. 
A particularly interesting case is when the extension of $\rho_{A}$ is chosen to be pure. In this case we obtain an alternative expression for $D(\rho_A\|\sigma_A)$ that avoids taking a logarithm of $\rho_A$ (see~\Cref{rmk_large}).

\item For any two density matrices $\rho_{AB}$, $\sigma_A$, and $\alpha \in [\frac{1}{2},\infty]$, we show that
\begin{align} \label{eq_main_res_single_letter}
D_{\alpha,\mathbb{M}}(\rho_A \| \sigma_A) \overset{\textnormal{DPI}}{\leq} D_{\alpha,\mathbb{M}}(\rho_{AB}\| \cC_{AB}^{\sigma_A}) \leq D_{\alpha}(\rho_A \| \sigma_A) \, ,
\end{align}
where the measured R\'enyi relative entropy is defined as $D_{\alpha,\mathbb{M}}(\rho\| \sigma):=\sup_{\cM}D_{\alpha}(\cM(\rho)\| \cM(\sigma))$, with the supremum over all POVM channels $\cM$~\cite{donald86,PH91,FBT15}.\footnote{The DPI for the R\'enyi relative entropy has been proven in~\cite[Theorem~1]{LF13}.}
For $\alpha=1$ we also find an explicit expression for a density matrix $\sigma_{AB} \in \cC_{AB}^{\sigma_A}$ satisfying~\cref{eq_main_res_single_letter} (see~\cref{rmk_explicit_extension}).
We call~\cref{eq_main_res_single_letter} a \textbf{measured Uhlmann's theorem} for relative entropies and refer to~\cref{thm_single_letter_Uhlmann_generalization,thm_key} for the technical statements.

\end{enumerate}

\paragraph{The classical case.}
In case $\rho_{AB}$ and $\sigma_A$ represent classical distributions, it is straightforward to see that~\cref{eq_main_res_single_letter,eq_main_res} are true. 
Let $\rho_{AB} = \sum_{a \in \cA,b \in \cB} p(a,b) \proj{a}_A \otimes \proj{b}_B$ and $\sigma_A = \sum_{a\in \cA} q(a) \proj{a}_A$ be classical states, where $p \in \Delta_{\cA \times \cB}$ and $q \in \Delta_{\cA} $ are probability distributions, i.e., they are specified by elements of the corresponding probability simplex \smash{$\Delta_{\cX}:=\{x \in \R^{|\cX|}: x_i \geq 0 \, ,$} $\sum_{i=1}^{|\cX|} x_i =1 \}$. 
The DPI for the R\'enyi relative entropy~\cite[Theorem~4.18]{marco_book} implies 
\begin{align} \label{eq_DPI2}
D_{\alpha}(\rho_{A}\| \sigma_{A}) \overset{\textnormal{DPI}}{\leq} D_{\alpha}(\rho_{AB}\| \cC^{\sigma_A}_{AB} )  \, .
\end{align}
For the other direction, note that
\begin{align}
D_{\alpha}(\rho_{AB}\| \cC^{\sigma_A}_{AB} )
&= \min_{\bar \sigma_{AB} \in \cC^{\sigma_A}_{AB}} D_{\alpha}(\rho_{AB}\| \bar \sigma_{AB}  )  \\
&\leq \frac{1}{\alpha-1} \log \left(   \sum_{a,b: p(a)>0} p(a,b)^{\alpha} \Big(p(a,b) \frac{q(a)}{p(a)}\Big)^{1-\alpha} \right) \\
&= D_{\alpha}(\rho_A \| \sigma_A) \label{eq_classical-1}\, ,
\end{align}
where the inequality step uses that
\begin{align}
\bar \sigma_{AB} = \sum_{a \in \cA,b \in \cB} \bar q(a,b) \proj{a}_A \otimes \proj{b}_B \quad \textnormal{for} \quad \bar q(a,b) :=  \begin{cases}  p(a,b) \frac{q(a)}{p(a)}, \  &\mathrm{if}\  p(a) > 0\\ p(b)q(a), \ &\mathrm{else} \end{cases}
\end{align}
fulfills $\bar \sigma_{AB} \in \cC^{\sigma_A}_{AB}$ because $\bar q(a,b)$ is a probability distribution that satisfies $\bar q(a) = q(a)$. This proves~\eqref{eq_main_res}, where the regularization is obsolete. In fact, combining this argument with~\cref{eq_Da_regul}, the DPI and the additivity under the tensor product for the R\'enyi relative entropy, it follows that $D_{\alpha}^{\infty}(\rho_{AB}\| \cC^{\sigma_A}_{AB} ) = D_{\alpha}(\rho_{AB}\| \cC^{\sigma_A}_{AB} )$. 
Furthermore, applying the same argument as above to $D_{\alpha,\mathbb{M}}(\rho_{AB}\| \cC^{\sigma_A}_{AB} )$ shows that \cref{eq_main_res_single_letter} is also true, where the two inequalities hold with equality since $D_{\alpha,\mathbb{M}}(\rho_A \| \sigma_A) = D_{\alpha}(\rho_A \| \sigma_A)$ in this case.
Finally, note that this classical proof can be applied to the quantum setting where $\rho_{AB}$, $\rho_{A}\otimes \id_B$, and $\sigma_{A}\otimes \id_B$ all commute, as observed in~\cite[Proposition~B.1]{GEAT_24}.


\paragraph{Paper organization.}
We begin with the preliminaries and notation in~\cref{sec_preliminaries}. 
We then prove our results in~\cref{sec_main_renyi} for R\'enyi relative entropies in the range $\alpha \in [\frac{1}{2},1) \cup (1,\infty]$. The proof idea is to lift the classical setting explained above to the non-commutative quantum case. To do so, we use the asymptotic spectral pinching method. In addition, we use a de Finetti reduction from~\cite{FR14}. In \Cref{sec_main}, the results are extended to $\alpha =1$.

\section{Preliminaries} \label{sec_preliminaries}
\subsection{Notation}
For two Hermitian matrices $X$ and $Y$, we say that $X \leq Y$ if $Y - X$ is positive semidefinite. We write $X>0$ and $X\geq 0$ for $X$ being a positive definite or positive semidefinite matrix, respectively.
The commutator is denoted by $[X,Y]:=XY-YX$.
Let $\St(A):=\{X \geq 0: \tr[X]=1\}$ denote the set of density matrices on a Hilbert space $A$ with dimension $d_A$.
For $n \in \N$ and a fixed $\sigma_A \in \St(A)$, we introduce the simplified notation

\begin{align} \label{eq_set_C_n}
\cC^{\sigma_A}_{AB,n}:= \cC^{\sigma^{\otimes n}_A}_{A_1^nB_1^n}  = \{ \sigma'_{A_1^n B_1^n} \in \St(A_1^n \otimes B_1^n): \sigma'_{A_1^n}=\sigma_A^{\otimes n} \} 
\end{align}
for the set of all extensions of $\sigma_A^{\otimes n}$ on the Hilbert space $B_1^n$.
For $n=1$, we drop the subscript $n$ (see~\cref{eq_set_C}). 
For any set of operators $\mathcal{L}$ on a Hilbert space $C$, the \emph{support function} is defined as
\begin{align} \label{eq_support_function}
h_{\mathcal{L}}(X) := \sup_{\sigma \in \mathcal{L}} \tr[X \sigma] \, ,
\end{align}
where $X$ is an operator on the Hilbert space $C$.
We write $\poly(n)$ for terms of order at most polynomial in $n$.


\subsection{Entropic quantities}
For an arbitrary divergence $\mathbb{D}$ between two density matrices $\rho,\sigma \in \St(A)$ we define a \emph{measured} version as
\begin{align} \label{eq_measured_divergence}
\mathbb{D}_{\mathbb{M}}(\rho \| \sigma) := \sup_{\cM} \mathbb{D}\big(\cM(\rho)\| \cM(\sigma)\big) \, ,
\end{align}
where the supremum is over all POVM channels $\cM$~\cite{donald86,PH91,FBT15}. For a set $\cS \subset \St(A)$ we use the notation introduced in~\cref{eq_def_Ua} to write 
\begin{align} \label{eq_divergence_set}
\mathbb{D}(\rho \| \cS ):= \inf_{\sigma \in \cS} \mathbb{D}(\rho \| \sigma) \, .
\end{align}
For a sequence of sets $\cS_n \subset \St(A_1^n)$ the regularization of this quantity is given by
\begin{align}
\mathbb{D}^{\infty}(\rho \| \cS ):= \lim_{n \to \infty} \frac{1}{n}  \mathbb{D}(\rho^{\otimes n} \| \cS_n ) \, ,
\end{align}
whenever the limit exists.

A popular divergence is the (sandwiched) R\'enyi relative entropy~\cite{MLDSFT13,WWY14} which is defined for $\alpha \in [\frac{1}{2},1) \cup (1,\infty)$ as
\begin{align} \label{eq_sandwiched}
D_{\alpha}(\rho \| \sigma)
:= \frac{1}{\alpha -1 } \log Q_{\alpha}(\rho\| \sigma)  \qquad \textnormal{for} \qquad Q_{\alpha}(\rho\| \sigma):=\tr[ (\sigma^{\frac{1-\alpha}{2\alpha}} \rho \sigma^{\frac{1-\alpha}{2\alpha}})^{\alpha}] \, ,
\end{align}
if the support of $\rho$ is included in the support of $\sigma$ or $\alpha<1$, and $+\infty$ otherwise.
For $\alpha=1$ and $\alpha=\infty$ we define $D_{1}=D$ and $D_{\infty}=D_{\max}$, respectively, where $D$ is the relative entropy defined in~\cref{eq_def_D} below (see \cref{lem_limits,lem_limits_infty}).
Following~\cref{eq_measured_divergence}, its measured version is denoted by $D_{\alpha,\mathbb{M}}(\rho\| \sigma)= \frac{1}{\alpha -1 } \log Q_{\alpha,\mathbb{M}}(\rho\| \sigma)$. 
This quantity features a variational formula
\begin{align} 
Q_{\alpha,\mathbb{M}}(\rho \| \sigma)
\overset{\textnormal{\cite[Lemma 3]{FBT15} }}&{=} \left \lbrace  \begin{array}{ll}
\inf_{\tau>0} \alpha\, \tr[\rho \tau^{\frac{\alpha-1}{\alpha}}]+(1-\alpha)\,\tr[\sigma \tau] & \textnormal{ if } \alpha \in [\frac{1}{2},1)\\
\sup_{\tau>0} \alpha \,\tr[\rho \tau^{\frac{\alpha-1}{\alpha}}]+(1-\alpha)\,\tr[\sigma \tau] & \textnormal{ if } \alpha \in (1,\infty)
\end{array}
\right. \label{eq_varFormula_DM_a} \\
\overset{\textnormal{\cite[Lemma 3]{FBT15} }}&{=}  \left \lbrace  \begin{array}{ll}
\inf_{\tau>0}  \tr[\rho \tau]^{\alpha}\,\tr[\sigma \tau^{\frac{\alpha}{\alpha-1}}]^{1-\alpha} & \textnormal{ if } \alpha \in [\frac{1}{2},1)\\
\sup_{\tau>0}  \tr[\rho \tau^{\frac{\alpha-1}{\alpha}}]^{\alpha}\,\tr[\sigma \tau]^{1-\alpha} & \textnormal{ if } \alpha \in (1,\infty) \, ,
\end{array}
\right.   \label{eq_varFormula_DM_a_2}
\end{align}
where the suprema and infima are achieved when $\rho$ and $\sigma$ have full rank.
Following our notation we have
\begin{align} \label{eq_def_UM_a}
D_{\alpha,\mathbb{M}}(\rho_{AB}\|\cC^{\sigma_A}_{AB})= \min_{\sigma_{AB} \in \cC^{\sigma_A}_{AB}}D_{\alpha,\mathbb{M}}(\rho_{AB}\|\sigma_{AB}) \, ,
\end{align}
and its regularization\footnote{\Cref{lem_remove_M_renyi} ensures that the regularized quantity $D^{\infty}_{\alpha,\mathbb{M}}(\rho_{AB}\|\cC^{\sigma_A}_{AB})$ is well-defined, i.e., the limit in~\cref{eq_regul_M} exists. However, it is unclear if the limit can be replaced with an infimum.} 
\begin{align} \label{eq_regul_M}
D^{\infty}_{\alpha,\mathbb{M}}(\rho_{AB}\|\cC^{\sigma_A}_{AB})=\lim_{n\to \infty} \frac{1}{n} D_{\alpha,\mathbb{M}}(\rho_{AB}^{\otimes n}\|\cC^{\sigma_A}_{AB,n})\,.
\end{align}
The DPI~\cite{marco_book} 
for the R\'enyi relative entropy implies~\cite[Theorem~6]{FBT15}\footnote{If $\rho >0$ and $\sigma >0$, then we also have that $D_{\alpha,\mathbb{M}}(\rho \| \sigma) = D_{\alpha}(\rho \| \sigma)$ implies $[\rho,\sigma]=0$ due to the Araki-Lieb-Thirring inequality~\cite{araki90,Lieb1991}, see~\cite[Theorem~6]{FBT15}.} 
\begin{align} \label{eq_euality_DM_vs_D}
D_{\alpha,\mathbb{M}}(\rho \| \sigma) \leq D_{\alpha}(\rho \| \sigma) \quad \textnormal{with equality if }\ [\rho,\sigma]=0 \, ,
\end{align}
and consequently
\begin{align} \label{eq_Ua_vs_Ua_meas}
D_{\alpha,\mathbb{M}}(\rho_{AB}\| \cC^{\sigma_A}_{AB}) \overset{\textnormal{\Cshref{eq_euality_DM_vs_D}}}{\leq} D_{\alpha}(\rho_{AB}\| \cC^{\sigma_A}_{AB}) \, .
\end{align}

For $\alpha=1$ we can define the same quantities using the relative entropy instead of the R\'enyi relative entropy and we drop the $\alpha$-subscript. The relative entropy is defined as
\begin{align} \label{eq_def_D}
D(\rho \| \sigma):=\tr[\rho (\log \rho - \log \sigma)] \, ,
\end{align}
if the support of $\rho$ is included in the support of $\sigma$, and $+\infty$ otherwise. It features a variational formula~\cite{Petz88_var,Sutter_book}
\begin{align}   \label{eq_varFormula_D}
D(\rho \| \sigma)
=  \sup_{\tau >0}\big \{ \tr[\rho \log \tau] - \log \tr[\ee^{\log \sigma + \log \tau} ] \big \} \, ,
\end{align}
where the supremum is achieved when $\rho$ and $\sigma$ have full rank.
The \emph{measured relative entropy}, defined according to~\cref{eq_measured_divergence}, also has a variational formula
\begin{align} 
\MD(\rho \| \sigma)
\overset{\textnormal{\cite[Lemma 1]{FBT15} }}&{=}  \sup_{\tau >0}\{ \tr[\rho \log \tau] - \log \tr[\sigma \tau ] \}   \label{eq_varFormula_DM} \\
\overset{\textnormal{\cite[Lemma 1]{FBT15} }}&{=}  \sup_{\tau >0}\{ \tr[\rho \log \tau] +1 -  \tr[\sigma \tau ] \} \, ,   \label{eq_varFormula_DM2}
\end{align}
where the suprema are achieved when $\rho$ and $\sigma$ have full rank.
The DPI~\cite{marco_book} for the relative entropy implies~\cite[Proposition~5]{FBT15}\footnote{If $\rho >0$ and $\sigma >0$, then we also have that $D_{\mathbb{M}}(\rho \| \sigma) = D(\rho \| \sigma)$ implies $[\rho,\sigma]=0$ due to the Golden-Thompson inequality~\cite{golden65,thompson65}, see~\cite[Proposition~5]{FBT15}.}
\begin{align} \label{eq_euality_DM_vs_D_ONE}
D_{\mathbb{M}}(\rho \| \sigma) \leq D(\rho \| \sigma) \quad \textnormal{with equality if }\ [\rho,\sigma]=0 \, .
\end{align}
Analogously as above recall the quantity
\begin{align} \label{eq_def_UM}
D_{\mathbb{M}}(\rho_{AB}\| \cC^{\sigma_A}_{AB})= \min_{\sigma_{AB} \in \cC^{\sigma_A}_{AB}}D_{\mathbb{M}}(\rho_{AB}\|\sigma_{AB}) \, ,
\end{align}
and its regularized version\footnote{\Cref{lem_remove_M} ensures that the regularized quantity $D^{\infty}_{\mathbb{M}}(\rho_{AB}\| \cC^{\sigma_A}_{AB})$ is well-defined, i.e.~the limit in~\cref{eq_regul_M_1} exists. However, it is unclear if it can be replaced by an infimum.}
\begin{align} \label{eq_regul_M_1}
D^{\infty}_{\mathbb{M}}(\rho_{AB}\| \cC^{\sigma_A}_{AB} )=\lim_{n\to \infty} \frac{1}{n} D_{\mathbb{M}}(\rho_{AB}^{\otimes n}\| \cC^{\sigma_A}_{AB,n})\,.
\end{align}
The measured R\'enyi relative entropy is not additive under the tensor product. Instead, for any $\alpha \in [\frac{1}{2}, \infty]$, it satisfies
\begin{align} \label{eq_remove_meas}
\lim_{n \to \infty} \frac{1}{n} D_{\alpha,\mathbb{M}}(\rho^{\otimes n} \| \sigma^{\otimes n} ) \overset{\textnormal{\cite[Theorem~4.18]{marco_book}}}{=} D_{\alpha}(\rho \| \sigma) \, .
\end{align}

\subsection{Asymptotic spectral pinching} \label{sec_pinching}
 Let $H$ be a Hermitian matrix with spectral decomposition $H=\sum_{\lambda \in \spec(H)} \lambda \Pi_{\lambda}$, where $\spec(H)$ is the set of distinct eigenvalues of $H$, and $\Pi_{\lambda}$ denotes the projector onto the eigenspace of $\lambda$. The \emph{pinching map} with respect to $H$ is defined as
 \begin{align}
 \cP_{H} : X \mapsto \sum_{\lambda \in \spec(H)} \Pi_{\lambda} X \Pi_{\lambda}\, .
 \end{align}
 This map has various useful properties. For any Hermitian matrix $H$ and any positive semidefinite matrix $X$, we have: 
 \begin{enumerate}[(I)]
     \item Commutation: $[\cP_{H}(X),H]=0$~\cite[Lemma~3.5]{Sutter_book}. \label{it_pinch_commutation}
     \item Invariance: $\cP_{H}(H)=H$, which follows by definition of the pinching map. \label{it_pinch_invariance}
     \item Pinching inequality: $X \leq \cP_{H}(X) |\spec(H)|$~\cite[Lemma~3.5]{Sutter_book}. \label{it_pinch_inequality}
     \item Commutation of pinching maps: If $[H_1,H_2]=0$ then $\cP_{H_1}\circ \cP_{H_2} = \cP_{H_2}\circ \cP_{H_1}$~\cite[Lemma~2.7]{GEAT_24}. \label{it_pinch_commutation_maps}
     \item Partial trace: $\tr_{B}[\cP_{H_A \otimes \id_B}(X_{AB})] = \cP_{H_A}(X_A)$~\cite[Lemma~2.7]{GEAT_24}. \label{it_pinch_partial_trace}
     \item Let $Y,Z \geq 0$ such that $[H,Y]=[H,Z]=0$, then $\cP_{H}(YXZ)=Y \cP_{H}(X) Z$.\footnote{This can be seen for example by the integral representation of pinching maps~\cite[Lemma~3.4]{Sutter_book}.} \label{it_pinch_move_when_commute} 
     \item The pinching map is a trace-preserving and completely positive operation~\cite[Exercise~3.3]{Sutter_book}. \label{it_pinch_CPTP}
 \end{enumerate}
We refer to~\cite{hayashi_book2} and~\cite[Section~3.1]{Sutter_book} for more information on pinching maps.


\subsection{De Finetti reductions}
Quantum de Finetti theorems~\cite{deFinetti_CKMR_07,renner07,renner_phd} relate permutation-invariant states to convex mixtures of tensor power states. In their \emph{postselection} formulation~\cite{postsel09} they state that any permutation-invariant density matrix $\rho_n$ on $\cH^{\otimes n}$ with $d=\dim(\cH)$ satisfies
\begin{align} \label{eq_postselection}
\rho_n \leq (n+1)^{d^2} \int \sigma^{\otimes n} \di\sigma \, ,
\end{align}
where $\di\sigma$ is a universal probability measure over $\St(\cH)$. 
For this work, we are interested in a refined version of~\cref{eq_postselection}, where the state $\rho_n$ has an additional structure (in addition to being permutation-invariant), and we want this structure to be preserved by the measure $\di\sigma$. 
Such a result has been derived in~\cite[Corollary~3.2]{FR14}. For any $\sigma_A \in \St(A)$ there exists a probability measure $\di\sigma_{AB}$ on the set of extensions $\sigma_{AB} \in \St(A \otimes B)$ of $\sigma_A$ such that
\begin{align} \label{eq_postselection_FR}
\rho_{A_1^n B_1^n} \leq (n+1)^{d^2} \int \sigma_{AB}^{\otimes n} \, \di\sigma_{AB} \, ,
\end{align}
for any $n \in \N$, $d = d_A d_B^2$, and any permutation-invariant extension $\rho_{A_1^n B_1^n} \in \cC^{\sigma_A}_{AB,n}$.

\section{Uhlmann's theorem for R\'enyi relative entropies} \label{sec_main_renyi}
In this section we prove our main results for R\'enyi relative entropies in the range $\alpha \in [\frac{1}{2},1) \cup (1,\infty] $.
The case $\alpha = 1$ is postponed to~\cref{sec_main}.
The cases $\alpha=\frac{1}{2}$ and $\alpha = \infty$ are also special as they correspond to the min- and max-relative entropy.

\subsection{Regularized Uhlmann's theorem for R\'enyi relative entropies}
\begin{theorem}[Regularized Uhlmann's theorem] \label{thm_tony_new_range}
Let $\rho_{AB} \in \St(A \otimes B)$, $\sigma_A \in \St(A)$, and $\alpha \in [\frac{1}{2},1) \cup (1,\infty] $. Then
\begin{align} \label{eq_tony_uhlmann}
D_{\alpha}(\rho_{A}\|\sigma_{A})= D_{\alpha}^{\infty}(\rho_{AB} \| \cC^{\sigma_A}_{AB}) \, . 
\end{align}
\end{theorem}
\begin{proof}
For $\alpha \in (1,\infty)$ the assertion has been proven in~\cite[Lemma~3.3]{GEAT_24}. For $\alpha=\infty$ the assertion follows by combining~\cite[Lemma~3.3]{GEAT_24} with~\cite[Corollary A.2]{MH11}. 
It thus remains to prove~\eqref{eq_tony_uhlmann} for $\alpha \in [\frac{1}{2},1)$.
The direction
\begin{align}
D_{\alpha}(\rho_{A}\|\sigma_{A}) =\lim_{n \to \infty} \frac{1}{n} D_{\alpha}(\rho_{A}^{\otimes n}\|\sigma_{A}^{\otimes n})
\overset{\textnormal{DPI}}{\leq} D_{\alpha}^{\infty}(\rho_{AB} \| \cC^{\sigma_A}_{AB})
\end{align}
follows from the additivity of the R\'enyi relative entropy under the tensor product, and the DPI~\cite[Theorem~4.18]{marco_book}.
To see the other direction, define for $n \in \N$ the following states
\begin{align} \label{eq_def_states_pinch}
\rho'_{A_1^n B_1^n}:=\cP_{\sigma_A^{\otimes n} \otimes \id_{B_1^n}}(\rho_{AB}^{\otimes n}) \qquad \textnormal{and} \qquad \bar \rho_{A_1^n B_1^n}:=\cP_{\rho'_{A_1^n} \otimes \id_{B_1^n}}(\rho'_{A_1^n B_1^n}) \, .
\end{align}
Consider the trace-preserving and completely positive map\footnote{In case $\bar \rho_{A_1^n}$ does not have full support, we only take the inverse on the support of $\bar \rho_{A_1^n}$.}
\begin{align}
\cE(X_{A_{1}^n}) := \bar \rho_{A_1^n B_1^n}^{\frac{1}{2}} \bar \rho_{A_1^n}^{-\frac{1}{2}} X_{A_{1}^n} \bar \rho_{A_1^n}^{-\frac{1}{2}} \bar \rho_{A_1^n B_1^n}^{\frac{1}{2}} 
\end{align}
and define
\begin{align} \label{eq_bar_sigma}
\bar \sigma_{A_1^n B_1^n}:= \cE(\sigma_{A}^{\otimes n}) \, .
\end{align}

We next summarize a few facts that will be used in this proof. They follow from the properties of pinching maps summarized in~\cref{sec_pinching}. 
\begin{enumerate}[(i)]
\item $\bar \rho_{A_1^n}\overset{\textnormal{\Cshref{eq_def_states_pinch}}}{=}\tr_{B_1^n}[\cP_{\rho'_{A_1^n} \otimes \id_{B_1^n}}(\rho'_{A_1^n B_1^n})]\overset{\eqref{it_pinch_partial_trace}}{=}\cP_{\rho'_{A_1^n}}(\rho'_{A_1^n})=\rho'_{A_1^n}$. \label{item_1}
\item \label{item_2} $[\bar \rho_{A_1^n B_1^n},\bar \rho_{A_1^n}\otimes \id_{B_1^n}]=0$ which follows by combining $[\bar \rho_{A_1^n B_1^n},\rho'_{A_1^n} \otimes \id_{B_1^n}]\overset{\eqref{it_pinch_commutation}}{=}0$ with~\eqref{item_1}.
\item \label{item_3} $[\bar \rho_{A_1^n B_1^n},\sigma_{A}^{\otimes n} \otimes \id_{B_1^n}]=0$, which follows because $[\rho'_{A_1^n B_1^n},\sigma_A^{\otimes n} \otimes \id_{B_1^n}]\overset{\eqref{it_pinch_commutation}}{=}0$ implies $[\rho'_{A_1^n},\sigma_A^{\otimes n}]=0$. Hence, we have  
\begin{align}
\bar\rho_{A_1^n B_1^n} 
\overset{\textnormal{\Cshref{eq_def_states_pinch}}}{=} \cP_{\rho'_{A_1^n} \otimes \id_{B_1^n}} \circ \cP_{\sigma_A^{\otimes n} \otimes \id_{B_1^n}}(\rho_{AB}^{\otimes n}) 
\overset{\eqref{it_pinch_commutation_maps}}{=}\cP_{\sigma_A^{\otimes n} \otimes \id_{B_1^n}} \circ \cP_{\rho'_{A_1^n} \otimes \id_{B_1^n}}(\rho_{AB}^{\otimes n}) \, .
\end{align}
\item\label{item_4} $[\bar \rho_{A_1^n},\sigma_{A}^{\otimes n}]=0$ which follows directly from~\eqref{item_3}.
\item\label{item_5}  $[\sigma_A^{\otimes n} \otimes \id_{B_1^n}, \bar \sigma_{A_1^n B_1^n} ]=0$, which follows from~\cref{eq_bar_sigma} together with~\eqref{item_3} and~\eqref{item_4}.
\item\label{item_6} $[\bar \rho_{A_1^n} \otimes \id_{B_1^n}, \bar \sigma_{A_1^n B_1^n} ]=0$, which follows from~\cref{eq_bar_sigma} together with~\eqref{item_2} and~\eqref{item_4}.
\end{enumerate}
Note that $\bar \sigma_{A_1^n B_1^n}$ defined in~\cref{eq_bar_sigma} is a feasible state in the sense that $\bar \sigma_{A_1^n B_1^n} \in \cC_{AB,n}^{\sigma_A}$ since
\begin{align}
\bar \sigma_{A_1^n} 
\overset{\textnormal{\Cshref{eq_bar_sigma}}}&{=} \tr_{B_1^n}[\bar \rho_{A_1^n B_1^n}^{\frac{1}{2}} \bar \rho_{A_1^n}^{-\frac{1}{2}} \sigma_{A}^{\otimes n} \bar \rho_{A_1^n}^{-\frac{1}{2}} \bar \rho_{A_1^n B_1^n}^{\frac{1}{2}}] \\
\overset{\textnormal{\eqref{item_4}}}&{=} \tr_{B_1^n}[\bar \rho_{A_1^n B_1^n}^{\frac{1}{2}} \bar \rho_{A_1^n}^{-1} \sigma_{A}^{\otimes n}\bar \rho_{A_1^n B_1^n}^{\frac{1}{2}}] \\
\overset{\textnormal{\eqref{item_2},\eqref{item_3}}}&{=} \tr_{B_1^n}[\bar \rho_{A_1^n B_1^n} \bar \rho_{A_1^n}^{-1} \sigma_{A}^{\otimes n}] \\
&=\sigma_A^{\otimes n} \, . \label{eq_bar_sigma_feasible}
\end{align}

We next observe that
\begin{align}
\frac{1}{n}D_{\alpha}(\bar \rho_{A_1^n B_1^n} \| \bar \sigma_{A_1^n B_1^n} )
&= \frac{1}{n}D_{\alpha}\big( \cE( \bar \rho_{A_1^n}) \| \cE( \sigma_{A}^{\otimes n}) \big) \\
\overset{\textnormal{DPI}}&{\leq}\frac{1}{n}D_{\alpha}(\bar \rho_{A_1^n} \| \sigma_{A}^{\otimes n}) \\
\overset{\textnormal{\eqref{item_1}}}&{=}\frac{1}{n}D_{\alpha}(\rho'_{A_1^n} \| \sigma_{A}^{\otimes n}) \\
\overset{\textnormal{\Cshref{eq_def_states_pinch}},\,\eqref{it_pinch_partial_trace}}&{=}\frac{1}{n}D_{\alpha}\big( \cP_{\sigma_A^{\otimes n}} (\rho_{A}^{\otimes n}) \| \sigma_{A}^{\otimes n} \big) \\
\overset{\eqref{it_pinch_invariance}}&{=}\frac{1}{n}D_{\alpha}\big( \cP_{\sigma_A^{\otimes n}} (\rho_{A}^{\otimes n}) \| \cP_{\sigma_A^{\otimes n}} ( \sigma_{A}^{\otimes n}) \big) \\
\overset{\textnormal{DPI}}&{\leq} \frac{1}{n} D_{\alpha}(\rho_{A}^{\otimes n} \|  \sigma_{A}^{\otimes n}) \\
&=D_{\alpha}(\rho_{A} \|  \sigma_{A}) \, . \label{eq_alpha_step_1_ds}
\end{align}
Applying twice the pinching inequality yields
\begin{align} \label{eq_pinching_in_pf}
\rho_{AB}^{\otimes n} \overset{\eqref{it_pinch_inequality}}{\leq} |\spec(\sigma_A^{\otimes n})| |\spec(\rho'_{A_1^n})| \bar \rho_{A_1^n B_1^n} 
\overset{\textnormal{\cite[Lemma~2.8 \& Corollary~2.9]{GEAT_24}}}{=} \poly(n)\bar \rho_{A_1^n B_1^n} \, .
\end{align}
Using that $t \mapsto t^{\alpha -1}$ for $\alpha \in [\frac{1}{2},1)$ is operator anti-monotone~\cite[Table~2.2]{Sutter_book}, and the commutation relations~\eqref{item_1}-\eqref{item_6} together with~\cref{fact_commute_3_matrices} and~\ref{fact_commute_power} give 
\begin{align}
&\hspace{-5mm}Q_{\alpha}(\rho_{AB}^{\otimes n} \|  \bar \sigma_{A_1^n B_1^n} )\nonumber \\
&=\tr\Big[\Big(\bar \sigma_{A_1^n B_1^n}^{\frac{1-\alpha}{2\alpha}}\, \rho_{AB}^{\otimes n} \, \bar \sigma_{A_1^n B_1^n}^{\frac{1-\alpha}{2\alpha}}\Big)^{\alpha-1}\ \, \bar \sigma_{A_1^n B_1^n}^{\frac{1-\alpha}{2\alpha}}\, \rho_{AB}^{\otimes n}\, \bar \sigma_{A_1^n B_1^n}^{\frac{1-\alpha}{2\alpha}}\Big] \\
\overset{\textnormal{\Cshref{eq_pinching_in_pf}}}&{\geq} \tr\Big[\Big(\bar \sigma_{A_1^n B_1^n}^{\frac{1-\alpha}{2\alpha}}\, \bar \rho_{A_1^n B_1^n }\, \bar \sigma_{A_1^n B_1^n}^{\frac{1-\alpha}{2\alpha}}\Big)^{\alpha-1}\ \, \bar \sigma_{A_1^n B_1^n}^{\frac{1-\alpha}{2\alpha}}\, \rho_{AB}^{\otimes n}\, \bar \sigma_{A_1^n B_1^n}^{\frac{1-\alpha}{2\alpha}}\Big] \poly(n)^{\alpha -1} \\
\overset{\eqref{it_pinch_CPTP}}&{=}\!\!\tr\Big[\!\cP_{\rho'_{A_1^n} \otimes \id_{B_1^n}}\!\! \circ \! \cP_{\sigma_A^{\otimes n} \otimes \id_{B_1^n}}\Big(\!\big(\bar \sigma_{A_1^n B_1^n}^{\frac{1-\alpha}{2\alpha}}\, \bar \rho_{A_1^n B_1^n }\, \bar \sigma_{A_1^n B_1^n}^{\frac{1-\alpha}{2\alpha}}\big)^{\alpha-1} \bar \sigma_{A_1^n B_1^n}^{\frac{1-\alpha}{2\alpha}}\, \rho_{AB}^{\otimes n}\, \bar \sigma_{A_1^n B_1^n}^{\frac{1-\alpha}{2\alpha}}\!\Big)\!\Big] \poly(n)^{\alpha -1}  \\
\overset{\eqref{item_3},\eqref{item_5},\eqref{it_pinch_move_when_commute},\eqref{fact_commute_3_matrices},\eqref{fact_commute_power}}&{=}\!\!\tr\Big[\cP_{\rho'_{A_1^n} \otimes \id_{B_1^n}}\!\Big(\!\big(\bar \sigma_{A_1^n B_1^n}^{\frac{1-\alpha}{2\alpha}}\, \bar \rho_{A_1^n B_1^n }\, \bar \sigma_{A_1^n B_1^n}^{\frac{1-\alpha}{2\alpha}}\big)^{\alpha-1} \ \bar \sigma_{A_1^n B_1^n}^{\frac{1-\alpha}{2\alpha}}\, \rho'_{A_1^n B_1^n}\, \bar \sigma_{A_1^n B_1^n}^{\frac{1-\alpha}{2\alpha}}\Big)\Big] \poly(n)^{\alpha -1}  \\
\overset{\eqref{item_1},\eqref{item_2},\eqref{item_6},\eqref{it_pinch_move_when_commute},\eqref{fact_commute_3_matrices},\eqref{fact_commute_power}}&{=} \tr\Big[\big(\bar \sigma_{A_1^n B_1^n}^{\frac{1-\alpha}{2\alpha}}\, \bar \rho_{A_1^n B_1^n } \,\bar \sigma_{A_1^n B_1^n}^{\frac{1-\alpha}{2\alpha}}\big)^{\alpha-1}\  \bar \sigma_{A_1^n B_1^n}^{\frac{1-\alpha}{2\alpha}}\, \bar \rho_{A_1^n B_1^n}\, \bar \sigma_{A_1^n B_1^n}^{\frac{1-\alpha}{2\alpha}}\Big] \poly(n)^{\alpha -1} \\
&=Q_{\alpha}(\bar \rho_{A_1^n B_1^n } \|  \bar \sigma_{A_1^n B_1^n})\, \poly(n)^{\alpha -1} \, .
\end{align}
This can be rewritten as
\begin{align}
\frac{1}{n}D_{\alpha}(\rho_{AB}^{\otimes n} \|  \bar \sigma_{A_1^n B_1^n})
\leq \frac{1}{n}D_{\alpha}(\bar \rho_{A_1^n B_1^n } \|  \bar \sigma_{A_1^n B_1^n}) + \frac{o(n)}{n} \, . \label{eq_alpha_step_2_ds}
\end{align}
We conclude the proof by noting
\begin{align}
\frac{1}{n}D_{\alpha}(\rho_{AB}^{\otimes n} \|  \cC^{\sigma_A}_{AB,n})
\overset{\textnormal{\Cshref{eq_bar_sigma_feasible}}}&{\leq} \frac{1}{n}D_{\alpha}(\rho_{AB}^{\otimes n} \|  \bar \sigma_{A_1^n B_1^n})\\
\overset{\textnormal{\Cshref{eq_alpha_step_2_ds}}}&{\leq} \frac{1}{n}D_{\alpha}(\bar \rho_{A_1^n B_1^n } \|  \bar \sigma_{A_1^n B_1^n}) + \frac{o(n)}{n} \\
\overset{\textnormal{\Cshref{eq_alpha_step_1_ds}}}&{\leq}D_{\alpha}(\rho_{A} \|  \sigma_{A}) + \frac{o(n)}{n} \, .
\end{align}
Taking the limit $n \to \infty$ proves the assertion.
\end{proof}
\subsection{Measured Uhlmann's theorem for R\'enyi relative entropies}
We next prove a single-letter version of Uhlmann's theorem for measured R\'enyi relative entropies (see~\cref{thm_single_letter_Uhlmann_generalization}). To do so, we need three preparatory lemmas which are discussed next.  
\begin{lemma} \label{lem_remove_M_renyi}
 Let $\rho_{AB} \in \St(A\otimes B)$, $\sigma_A \in \St(A)$, and $\alpha \in [\frac{1}{2},1) \cup (1,\infty)$. Then,
\begin{align} \label{eq_prop_remove_M_renyi}
 D^{\infty}_{\alpha,\mathbb{M}}(\rho_{AB} \| \cC^{\sigma_A}_{AB})=  D^{\infty}_{\alpha}(\rho_{AB} \| \cC^{\sigma_A}_{AB} )  \, . 
\end{align}
\end{lemma}
\begin{proof}
The DPI for the R\'enyi relative entropy~\cite[Theorem~4.18]{marco_book} implies that for any $n \in \N$
 \begin{align}
  \frac{1}{n} D_{\alpha,\mathbb{M}}(\rho_{AB}^{\otimes n} \| \cC^{\sigma_A}_{AB,n} ) \overset{\textnormal{\Cshref{eq_Ua_vs_Ua_meas}}}{\leq}   \frac{1}{n} D_{\alpha}(\rho_{AB}^{\otimes n} \| \cC^{\sigma_A}_{AB,n})  \, .
 \end{align}
 Taking the limit $n \to \infty$ yields
 \begin{align}
 D^{\infty}_{\alpha,\mathbb{M}}(\rho_{AB}\| \cC^{\sigma_A}_{AB} ) \leq D_{\alpha}^{\infty}(\rho_{AB}\| \cC^{\sigma_A}_{AB})\, .
 \end{align}
 
 Hence, it remains to prove the other direction. For any $n \in \N$ we denote by $\cS_n$ the set of permutations on $n$ elements.
 The unitary invariance and concavity/convexity properties of the measured R\'enyi relative entropy imply that for any $\bar \sigma_{A_1^n B_1^n} \in \cC^{\sigma_A}_{AB,n}$, we have for $\alpha \in [\frac{1}{2},1)$ 
 \begin{align} 
 Q_{\alpha,\mathbb{M}}(\rho_{AB}^{\otimes n} \| \bar \sigma_{A_1^n B_1^n}) 
 \overset{\textnormal{\Cshref{fact_UI}}}&{=} \frac{1}{n!} \sum_{\pi \in \cS_n}   Q_{\alpha,\mathbb{M}}\big(\pi \rho_{AB}^{\otimes n} \pi^\dagger \|  \pi \bar \sigma_{A_1^n B_1^n} \pi^\dagger \big) \\
 &=\frac{1}{n!} \sum_{\pi \in \cS_n}   Q_{\alpha,\mathbb{M}}\big(\rho^{\otimes n} \|  \pi \bar \sigma_{A_1^n B_1^n} \pi^\dagger \big) \\
 \overset{\textnormal{\Cshref{fact_convex_concave}}}&{\leq} Q_{\alpha,\mathbb{M}}  \Big( \rho^{\otimes n} \Big\|   \frac{1}{n!} \sum_{\pi \in \cS_n}   \pi \bar \sigma_{A_1^n B_1^n} \pi^\dagger  \Big) \, .  \label{eq_opt_PI_renyi} 
 \end{align}
By applying the same arguments for $\alpha \in (1,\infty)$ we find
\begin{align}
 Q_{\alpha,\mathbb{M}}(\rho_{AB}^{\otimes n} \| \bar \sigma_{A_1^n B_1^n}) 
 \geq Q_{\alpha,\mathbb{M}}  \Big( \rho^{\otimes n} \Big\|   \frac{1}{n!} \sum_{\pi \in \cS_n}   \pi \bar \sigma_{A_1^n B_1^n} \pi^\dagger  \Big) \, .
\end{align}

Furthermore, since 
 \begin{align}
  \tr_{B_1^n}\Big[\frac{1}{n!} \sum_{\pi \in \cS_n}   \pi \bar \sigma_{A_1^n B_1^n} \pi^\dagger\Big]
  &=\frac{1}{n!}\sum_{\pi \in \cS_n} \tr_{B_1^n}[  \pi_{A} \otimes \pi_B \bar \sigma_{A_1^n B_1^n} \pi_A^\dagger \otimes \pi_B^\dagger] \\
  &=\frac{1}{n!}\sum_{\pi \in \cS_n}  \pi_{A} \tr_{B_1^n}[ \pi_B \bar \sigma_{A_1^n B_1^n}  \pi_B^\dagger] \pi_A^\dagger \\
  &=\frac{1}{n!}\sum_{\pi \in \cS_n}  \pi_{A} \tr_{B_1^n}[ \bar \sigma_{A_1^n B_1^n}] \pi_A^\dagger \\
    &=\frac{1}{n!}\sum_{\pi \in \cS_n}  \pi_{A} \sigma_A^{\otimes n}  \pi_A^\dagger \\
  &=\sigma_A^{\otimes n} \, , \label{eq_valid_extentision}
 \end{align}
 it follows that for any $\alpha \in [\frac{1}{2},1) \cup (1,\infty)$
 \begin{align} \label{eq_optimizer_Er_renyi}
  \hat \sigma_{A_1^n B_1^n} \in \argmin_{\bar \sigma_{A_1^n B_1^n} \in \cC^{\sigma_A}_{AB,n}}  D_{\alpha,\mathbb{M}}(\rho_{AB}^{\otimes n} \| \bar \sigma_{A_1^n B_1^n}) 
 \end{align}
can be assumed to be permutation-invariant. By~\cref{fact_asymptotic_DM_perminv}, it thus follows that
\begin{align}
\frac{1}{n} D_{\alpha,\mathbb{M}}(\rho_{AB}^{\otimes n} \| \cC^{\sigma_A}_{AB,n} ) 
\overset{\textnormal{\Cshref{eq_optimizer_Er_renyi}}}&{=} \frac{1}{n}  D_{\alpha,\mathbb{M}}(\rho_{AB}^{\otimes n} \| \hat \sigma_{A_1^n B_1^n}) \\
\overset{\textnormal{\Cshref{fact_asymptotic_DM_perminv}}}&{\geq} \frac{1}{n}  D_{\alpha}( \rho_{AB}^{\otimes n} \| \hat \sigma_{A_1^n B_1^n} ) + \frac{o(n)}{n} \\
&\geq \frac{1}{n} D_{\alpha}(\rho_{AB}^{\otimes n} \| \cC^{\sigma_A}_{AB,n} ) + \frac{o(n)}{n} \, .
\end{align}
 Taking the limit $n \to \infty$ yields
 \begin{align}
 D^{\infty}_{\alpha,\mathbb{M}}(\rho_{AB} \| \cC^{\sigma_A}_{AB} ) \geq D_{\alpha}^{\infty}(\rho_{AB} \| \cC^{\sigma_A}_{AB})\, .
 \end{align}
\end{proof}
Next, we show that the support function for the set $\cC^{\sigma_A}_{AB,n}$ defined in~\cref{eq_set_C_n} is approximately multiplicative under the tensor product.
\begin{lemma} \label{lem_additivity_support_function}
Let $\tau_{AB} \in \St(A\otimes B)$, $\sigma_A \in \St(A)$, and $\cC^{\sigma_A}_{AB,n}$ as defined in~\cref{eq_set_C_n}. Then,
\begin{align} \label{eq_support_additivity}
\frac{1}{n} \log h_{\cC^{\sigma_A}_{AB,n}}(\tau_{AB}^{\otimes n}) = \log h_{\cC^{\sigma_A}_{AB}}(\tau_{AB}) + \frac{o(n)}{n} \, , 
\end{align}
where $h_{\cC^{\sigma_A}_{AB,n}}$ denotes the support function of $\cC^{\sigma_A}_{AB,n}$ defined in~\cref{eq_support_function}.
\end{lemma}
\begin{proof}
We note that~\cref{eq_support_additivity} is a relaxed asymptotic version of the \emph{polar assumption} introduced in~\cite[Lemma~8]{FFF24}. 
One direction of~\cref{eq_support_additivity} is simple since\footnote{Since we consider finite dimensions, the compactness of the set $\cC^{\sigma_A}_{AB,n}$ implies that the support function achieves its supremum.}
\begin{align}
\frac{1}{n} \log h_{\cC^{\sigma_A}_{AB,n}}(\tau^{\otimes n})
= \frac{1}{n} \log \max_{\sigma_n \in \cC^{\sigma_A}_{AB,n} } \tr[\sigma_n \tau^{\otimes n}]
\geq  \frac{1}{n} \log \max_{\sigma \in \cC^{\sigma_A}_{AB}} \tr[\sigma^{\otimes n} \tau^{\otimes n}]
= \log h_{\cC^{\sigma_A}_{AB}}(\tau)  \, ,
\end{align}
where we used that $\sigma \in \cC^{\sigma_A}_{AB}$ implies $\sigma^{\otimes n} \in \cC^{\sigma_A}_{AB,n}$.
To see the other direction, note that the maximizer
\begin{align}
\bar \sigma_{A_1^n B_1^n} \in \argmax_{\sigma_n \in \cC^{\sigma_A}_{AB,n}} \tr[\tau^{\otimes n} \sigma_n] 
\end{align}
can be assumed without loss of generality to be permutation-invariant, because for any $\sigma_n \in \cC^{\sigma_A}_{AB,n}$,
\begin{align}
\tr[\tau^{\otimes n}  \sigma_n]
\overset{\textnormal{cyclicity of trace}}&{=} \frac{1}{n!} \sum_{\pi \in \cS_n} \tr[\pi \tau^{\otimes n} \pi^\dagger \pi \sigma_n \pi^\dagger] \\
&=\frac{1}{n!} \sum_{\pi \in \cS_n} \tr[\tau^{\otimes n}  \pi \sigma_n \pi^\dagger] \\
&= \tr\Big[\tau^{\otimes n} \frac{1}{n!} \sum_{\pi \in \cS_n} \pi \sigma_n \pi^\dagger \Big]   \, . 
\end{align}
In addition,~\cref{eq_valid_extentision} ensures that $\frac{1}{n!} \sum_{\pi \in \cS_n} \pi \sigma_n \pi^\dagger \in \cC^{\sigma_A}_{AB,n}$.
Hence, we can employ a quantum de Finetti reduction which yields for $d=\dim(A)\dim(B)^2$
\begin{align}
\frac{1}{n} \log h_{\cC^{\sigma_A}_{AB,n}}(\tau^{\otimes n})
&= \frac{1}{n} \log \max_{\sigma_n \in \cC^{\sigma_A}_{AB,n} } \tr[\sigma_n \tau^{\otimes n}] \\
\overset{\textnormal{\Cshref{eq_postselection_FR}}}&{\leq} \frac{1}{n} \log \left( (n+1)^{d^2} \int \tr[\sigma^{\otimes n} \tau^{\otimes n}] \di \sigma \right) \\
&=\frac{1}{n} \log \left( (n+1)^{d^2} \int (\tr[\sigma \tau])^n \di \sigma \right) \\
&\leq \log h_{\cC^{\sigma_A}_{AB}}(\tau) + \frac{o(n)}{n} \, ,
\end{align}
where the final inequality uses that $\tr[\sigma \tau] \leq h_{\cC^{\sigma_A}_{AB}}(\tau) $ for any $\sigma \in \cC^{\sigma_A}_{AB}$  and that $\di \sigma$ is a probability measure on $\cC^{\sigma_A}_{AB}$.
This proves~\cref{eq_support_additivity}.
\end{proof}

\begin{lemma} \label{lem_superadd_renyi}
 Let $\rho_{AB} \in \St(A\otimes B)$, $\sigma_A \in \St(A)$, and $\alpha \in [\frac{1}{2},1) \cup (1,\infty)$. Then,
\begin{align} \label{eq_prop_superadd_renyi}
 D^{\infty}_{\alpha,\mathbb{M}}(\rho_{AB}\| \cC^{\sigma_A}_{AB}) \geq  D_{\alpha,\mathbb{M}}(\rho_{AB}\| \cC^{\sigma_A}_{AB})  \, . 
\end{align}
\end{lemma}
\begin{proof}
We first prove the assertion in the range $\alpha \in [\frac{1}{2},1)$. For $n \in \N$ we have
\begin{align}
\frac{1}{n} D_{\alpha,\mathbb{M}}(\rho_{AB}^{\otimes n}\| \cC^{\sigma_A}_{AB,n})
=\frac{1}{n} \min_{\sigma_n \in \cC^{\sigma_A}_{AB,n}} D_{\alpha,\mathbb{M}}(\rho_{AB}^{\otimes n} \| \sigma_n)
=\frac{1}{n} \frac{1}{\alpha -1} \log \max_{\sigma_n \in \cC^{\sigma_A}_{AB,n}} Q_{\alpha,\mathbb{M}}(\rho_{AB}^{\otimes n} \| \sigma_n) \, . \label{eq_var_alpha_step1_ds}
\end{align}
Using the variational formula from~\cite[Lemma~3]{FBT15} we have
\begin{align}
\max_{\sigma_n \in \cC^{\sigma_A}_{AB,n}} Q_{\alpha,\mathbb{M}}(\rho^{\otimes n} \| \sigma_n)
\overset{\textnormal{\Cshref{eq_varFormula_DM_a}}}&{=}\max_{\sigma_n \in \cC^{\sigma_A}_{AB,n}} \inf_{\tau_n >0}\big\{ \alpha \tr[\rho^{\otimes n} \tau_n^{\frac{\alpha-1}{\alpha}}] + (1-\alpha) \tr[\sigma_n \tau_n] \big\} \\
&\leq \inf_{\tau_n >0}\big\{ \alpha \tr[\rho^{\otimes n} \tau_n^{\frac{\alpha-1}{\alpha}}] + (1-\alpha) h_{\cC^{\sigma_A}_{AB,n}}(\tau_n)\big\} \\
&\leq \inf_{\tau >0}\big\{ \alpha \tr[\rho^{\otimes n} (\tau^{\frac{\alpha-1}{\alpha}})^{\otimes n}] + (1-\alpha) h_{\cC^{\sigma_A}_{AB,n}}(\tau^{\otimes n})\big\}\\
\overset{\textnormal{\Cshref{lem_additivity_support_function}}}&{\leq}\inf_{\tau >0}\big\{ \alpha (\tr[\rho \tau^{\frac{\alpha-1}{\alpha}}])^n + (1-\alpha) h_{\cC}(\tau)^n\big\} \,\ee^{o(n)} \\
&=\inf_{\tau >0} \max_{\sigma \in \cC^{\sigma_A}_{AB}}\big\{ \alpha \tr[\rho \tau^{\frac{\alpha-1}{\alpha}}]^n + (1-\alpha) \tr[\sigma \tau]^n\big\} \,\ee^{o(n)} \\
\overset{\textnormal{Sion's minimax~\cite{Sion58}}}&{=} \max_{\sigma \in \cC^{\sigma_A}_{AB}}\inf_{\tau >0}  \big\{ \alpha \tr[\rho \tau^{\frac{\alpha-1}{\alpha}}]^n + (1-\alpha) \tr[\sigma \tau]^n\big\} \,\ee^{o(n)} \label{eq_before_diamond} \\
\overset{(\diamond)}&{\leq} \max_{\sigma \in \cC^{\sigma_A}_{AB}}\inf_{\tau >0}  \big\{ \tr[\rho \tau^{\frac{\alpha-1}{\alpha}}]^{n\alpha}  \tr[\sigma \tau]^{n(1-\alpha)}\big\} \,\ee^{o(n)} \label{eq_after_diamond} \\
&= \max_{\sigma \in \cC^{\sigma_A}_{AB}}\inf_{\tau >0}  \big\{ \tr[\rho \tau]^{n\alpha}  \tr[\sigma \tau^{{\frac{\alpha}{\alpha-1}}}]^{n(1-\alpha)}\big\} \,\ee^{o(n)}\\
\overset{\textnormal{\Cshref{eq_varFormula_DM_a_2}}}&{=} \max_{\sigma \in \cC^{\sigma_A}_{AB}}Q_{\alpha,\mathbb{M}}(\rho \| \sigma)^n \,\ee^{o(n)} \, , \label{eq_var_alpha_step2_ds}
\end{align}
where Sion's minimax theorem is applicable since the set of positive definite matrices and $\cC^{\sigma_A}_{AB}$ are both convex. In addition, the set $\cC^{\sigma_A}_{AB}$ is compact since we are in a finite-dimensional setting. Furthermore, the function $\tau \mapsto \alpha \tr[\rho \tau^{\frac{\alpha-1}{\alpha}}]^n + (1-\alpha) \tr[\sigma \tau]^n$ is convex~\cite[Table~2.2]{Sutter_book}\footnote{Recall that for $n\geq 1$ the function $h: x \mapsto x^n$ is convex and non-decreasing on $\R_+$. Hence, for a convex function $g$ the composition $f=h \circ g$ is convex~\cite[Section~3.2.4]{boyd_book}.} and $\sigma \mapsto \alpha \tr[\rho \tau^{\frac{\alpha-1}{\alpha}}]^n + (1-\alpha) \tr[\sigma \tau]^n$ is quasi-concave.\footnote{To see this, let $f(\sigma):=c + (1-\alpha) \tr[\sigma \tau]^n$ for $c\geq 0$ and note that for $t \in [0,1]$ we have $f(t \sigma_1 +(1-t)\sigma_2)= c +(1-\alpha) (t \tr[\sigma_1 \tau] +(1-t)\tr[\sigma_2 \tau])^n \geq c +(1-\alpha) \min\{  \tr[\sigma_1 \tau]^n ,\tr[\sigma_2 \tau]^n\} = \min\{f(\sigma_1) ,f(\sigma_2)\}$.}
To see the step $(\diamond)$ above, note that if $\tau >0$ is feasible in~\cref{eq_before_diamond} and $\lambda >0$, then $\lambda \tau >0$ is also feasible. Choosing \smash{$\lambda=\tr[\rho \tau^{\frac{\alpha-1}{\alpha}}]^\alpha \tr[\sigma \tau]^{-\alpha}$} justifies the inequality labeled with $(\diamond)$.\footnote{The inequality actually holds with equality due to the arithmetic-geometric mean inequality which ensures that $\beta x + (1-\beta)y \geq x^{\beta} y^{1-\beta}$ for $x,y \geq 0$ and $\beta \in [0,1]$. This argument has been introduced in~\cite[Proof of Lemma~3]{FBT15}.}
Putting everything together yields
\begin{align}
\frac{1}{n} D_{\alpha,\mathbb{M}}(\rho_{AB}^{\otimes n}\| \cC^{\sigma_A}_{AB,n})
\overset{\textnormal{\Cshref{eq_var_alpha_step1_ds}}}&{=} \frac{1}{n} \frac{1}{\alpha -1} \log \max_{\sigma_n \in \cC^{\sigma_A}_{AB,n}} Q_{\alpha,\mathbb{M}}(\rho_{AB}^{\otimes n} \| \sigma_n) \\
\overset{\textnormal{\Cshref{eq_var_alpha_step2_ds}}}&{\geq} D_{\alpha,\mathbb{M}}(\rho_{AB}\| \cC^{\sigma_A}_{AB}) + \frac{o(n)}{n} \, .
\end{align}
Taking the limit $n \to \infty$ proves the assertion of the lemma.

For $\alpha \in (1,\infty)$ the above proof technique can be reused with a few minor modifications.
For $n \in \N$ we have
\begin{align}
\frac{1}{n} D_{\alpha,\mathbb{M}}(\rho_{AB}^{\otimes n}\|\cC^{\sigma_A}_{AB,n} )
=\frac{1}{n} \min_{\sigma_n \in \cC^{\sigma_A}_{AB,n}} D_{\alpha,\mathbb{M}}(\rho_{AB}^{\otimes n} \| \sigma_n)
=\frac{1}{n} \frac{1}{\alpha -1} \log \min_{\sigma_n \in \cC^{\sigma_A}_{AB,n}} Q_{\alpha,\mathbb{M}}(\rho_{AB}^{\otimes n} \| \sigma_n) \, . \label{eq_var_alpha_step1_ds_a}
\end{align}
Using the variational formula from~\cite[Lemma~3]{FBT15} we have
\begin{align}
\min_{\sigma_n \in \cC^{\sigma_A}_{AB,n}} Q_{\alpha,\mathbb{M}}(\rho^{\otimes n} \| \sigma_n)
\overset{\textnormal{\Cshref{eq_varFormula_DM_a_2}}}&{=}\min_{\sigma_n \in \cC^{\sigma_A}_{AB,n}} \sup_{\tau_n >0}\big\{  \tr[\rho^{\otimes n} \tau_n^{\frac{\alpha-1}{\alpha}}]^{\alpha} \tr[\sigma_n \tau_n]^{1-\alpha} \big\} \\
&\geq \sup_{\tau_n >0}\big\{ \tr[\rho^{\otimes n} \tau_n^{\frac{\alpha-1}{\alpha}}]^{\alpha}  h_{\cC^{\sigma_A}_{AB,n}}(\tau_n)^{1-\alpha}\big\} \\
&\geq \sup_{\tau >0}\big\{  \tr[\rho^{\otimes n} (\tau^{\frac{\alpha-1}{\alpha}})^{\otimes n}]^{\alpha} h_{\cC^{\sigma_A}_{AB,n}}(\tau^{\otimes n})^{1-\alpha}\big\}\\
\overset{\textnormal{\Cshref{lem_additivity_support_function}}}&{\geq}\big(\sup_{\tau >0}\big\{  \tr[\rho \tau^{\frac{\alpha-1}{\alpha}}]^\alpha  h_{\cC}(\tau)^{1-\alpha}\big\} \big)^n \,\ee^{-o(n)} \\
&=\big(\sup_{\tau >0} \min_{\sigma \in \cC^{\sigma_A}_{AB}}\big\{\tr[\rho \tau^{\frac{\alpha-1}{\alpha}}]^{\alpha}  \tr[\sigma \tau]^{1-\alpha}\big\}\big)^n \,\ee^{-o(n)} \\
\overset{(\ddagger)}&{\geq} \big( \sup_{\tau >0}\min_{\sigma \in \cC^{\sigma_A}_{AB}}  \big\{ \alpha \tr[\rho \tau^{\frac{\alpha-1}{\alpha}}] + (1-\alpha) \tr[\sigma \tau]\big\}\big)^n \,\ee^{-o(n)}   \\
\overset{\textnormal{Sion's minimax~\cite{Sion58}}}&{=}\big( \min_{\sigma \in \cC^{\sigma_A}_{AB}}\sup_{\tau >0}  \big\{ \alpha \tr[\rho \tau^{\frac{\alpha-1}{\alpha}}] + (1-\alpha) \tr[\sigma \tau]\big\} \big)^n \,\ee^{-o(n)}  \\
\overset{\textnormal{\Cshref{eq_varFormula_DM_a}}}&{=} \min_{\sigma \in \cC^{\sigma_A}_{AB}}Q_{\alpha,\mathbb{M}}(\rho \| \sigma)^n \,\ee^{-o(n)} \, , \label{eq_var_alpha_step2_ds_a}
\end{align}
where the Sion's minimax step follows by analogous arguments as above in the proof for $\alpha \in [\frac{1}{2},1)$. The step $(\ddagger)$ follows from Bernoulli's inequality. More precisely for $x,y>0$ and $\alpha>1$ we have
\begin{align}
x^{\alpha} y^{1-\alpha} 
= y \left( 1+ \Big(\frac{x}{y} -1\Big) \right)^{\alpha}
\overset{\textnormal{Bernoulli's inequality}}{\geq} y \left(1+\alpha\Big(\frac{x}{y}-1\Big)\right)
= \alpha x +(1-\alpha)y \, .
\end{align}
Combining everything yields
\begin{align}
\frac{1}{n} D_{\alpha,\mathbb{M}}(\rho_{AB}^{\otimes n}\| \cC^{\sigma_A}_{AB,n})
\overset{\textnormal{\Cshref{eq_var_alpha_step1_ds_a}}}&{=} \frac{1}{n} \frac{1}{\alpha -1} \log \min_{\sigma_n \in \cC^{\sigma_A}_{AB,n}} Q_{\alpha,\mathbb{M}}(\rho_{AB}^{\otimes n} \| \sigma_n) \\
\overset{\textnormal{\Cshref{eq_var_alpha_step2_ds_a}}}&{\geq} D_{\alpha,\mathbb{M}}(\rho_{AB}\| \cC^{\sigma_A}_{AB} ) + \frac{o(n)}{n} \, .
\end{align}
Taking the limit $n \to \infty$ proves the assertion of the lemma.
\end{proof}

\begin{remark} \label{rmk_bounds_on_U_a}
A consequence of~\cref{lem_remove_M_renyi,lem_superadd_renyi} is that we can derive single-letter upper and lower bounds for the regularized quantity $D_{\alpha}^{\infty}(\rho_{AB}\| \cC^{\sigma_A}_{AB})$ with $\alpha \in [\frac{1}{2},1) \cup (1,\infty]$ because
\begin{align}
 D_{\alpha,\mathbb{M}}(\rho_{AB}\| \cC^{\sigma_A}_{AB}) 
 \overset{\textnormal{\Cshref{lem_superadd_renyi}}}{\leq} D^{\infty}_{\alpha,\mathbb{M}}(\rho_{AB}\| \cC^{\sigma_A}_{AB}) 
 \overset{\textnormal{\Cshref{lem_remove_M_renyi}}}{=}D_{\alpha}^{\infty}(\rho_{AB} \| \cC^{\sigma_A}_{AB})
 \leq  D_{\alpha}(\rho_{AB} \| \cC^{\sigma_A}_{AB}) \, .
\end{align}
Note that $ D_{\alpha,\mathbb{M}}(\rho_{AB}\| \cC^{\sigma_A}_{AB})$ as well as $D_{\alpha}(\rho_{AB}\| \cC^{\sigma_A}_{AB})$ are convex optimization problems.
\end{remark}\label{rmk_bounds_U_renyi}

We are now ready to state the main result of this subsection.
\begin{theorem}[Measured Uhlmann's theorem] \label{thm_single_letter_Uhlmann_generalization}
Let $\rho_{AB} \in \St(A \otimes B)$, $\sigma_A \in \St(A)$, and $\alpha \in [\frac{1}{2},1) \cup (1,\infty]$. Then
\begin{align} \label{eq_lemma_important_renyi}
D_{\alpha,\mathbb{M}}(\rho_A \| \sigma_A) \leq D_{\alpha,\mathbb{M}}(\rho_{AB}\|\cC^{\sigma_A}_{AB}) \leq D_{\alpha}(\rho_A \| \sigma_A) \, .
\end{align}
\end{theorem}
\begin{proof}
The first inequality in~\cref{eq_lemma_important_renyi} follows from the DPI for the measured R\'enyi relative entropy~\cite[Lemma~5]{RSB24}. Hence, it remains to justify the second inequality in~\cref{eq_lemma_important_renyi}. To do so, note that for $\alpha \in [\frac{1}{2},1) \cup (1,\infty)$ we have
\begin{align}
D_{\alpha,\mathbb{M}}(\rho_{AB}\| \cC^{\sigma_A}_{AB})
\overset{\textnormal{\Cshref{lem_superadd_renyi}}}{\leq} D^{\infty}_{\alpha,\mathbb{M}}(\rho_{AB} \| \cC^{\sigma_A}_{AB})
\overset{\textnormal{\Cshref{{lem_remove_M_renyi}}}}{=}D^{\infty}_{\alpha}(\rho_{AB}\| \cC^{\sigma_A}_{AB})
\overset{\textnormal{\Cshref{thm_tony_new_range}}}{=}D_{\alpha}(\rho_A \| \sigma_A) \, .
\end{align}
The case $\alpha =\infty$ then follows with the help of~\cref{lem_limits_infty}. 
\end{proof}
Unlike in Uhlmann's theorem (see~\cref{thm_uhlmann}), in the case where $\rho_{AB}$ is pure, the optimizer in the definition of $D_{\alpha,\mathbb{M}}(\rho_{AB}\| \cC^{\sigma_A}_{AB})$ is generally not pure~\cite[Appendix~H.1]{ernest_phd}.\footnote{In~\cite[Appendix~H.1]{ernest_phd} it is shown that the optizer is only pure for $\alpha=\frac{1}{2}$.} 
We refer to~\cref{app_counterexample_pure} for an example illustrating this.

\begin{remark}
\Cref{thm_single_letter_Uhlmann_generalization} implies a few known results:
\begin{enumerate}[(a)]
\item For $\alpha=\frac{1}{2}$ we recover Uhlmann's theorem in the form of~\cref{eq_Uhlmann_our_Not}, because the measured and non-measured min-relative entropies coincide~\cite{FBT15}.
\item For $\alpha=\infty$ we recover~\cref{eq_uhlmann_max}, because the measured and nonmeasured max-relative entropies coincide~\cite[Appendix A]{MO15}.
\item For $\alpha \to 1$ we recover the assertion of~\cref{thm_key} as explained in its proof.
\end{enumerate}
\end{remark}

\begin{remark} \label{rmk_min_max_simplify}
For the min-relative entropy ($\alpha=\frac{1}{2}$) and the max-relative entropy (\smash{$\alpha = \infty$}) some quantities simplify because the measured and non-measured quantities coincide~\cite[Appendix A]{MO15}.  In particular, for any $\rho_{AB} \in \St(A \otimes B)$, $\sigma_A \in \St(A)$ we have
\begin{align}
D^{\infty}_{\min,\mathbb{M}}(\rho_{AB} \| \cC^{\sigma_A}_{AB})
\!=\!D_{\min}^{\infty}(\rho_{AB} \| \cC^{\sigma_A}_{AB}) 
\!\overset{\textnormal{\Cshref{thm_tony_new_range},\ref{thm_single_letter_Uhlmann_generalization}}}{=}\! D_{\min}(\rho_{AB} \| \cC^{\sigma_A}_{AB})
\!=\! D_{\min,\mathbb{M}}(\rho_{AB} \| \cC^{\sigma_A}_{AB})
\end{align}
and
\begin{align}
D^{\infty}_{\max,\mathbb{M}}(\rho_{AB} \| \cC^{\sigma_A}_{AB})
\!=\!D_{\max}^{\infty}(\rho_{AB} \| \cC^{\sigma_A}_{AB}) 
\!\overset{\textnormal{\Cshref{thm_tony_new_range},\ref{thm_single_letter_Uhlmann_generalization}}}{=}\!  D_{\max}(\rho_{AB} \| \cC^{\sigma_A}_{AB})
\!=\! D_{\max,\mathbb{M}}(\rho_{AB} \| \cC^{\sigma_A}_{AB})
 \, .
\end{align}
\end{remark}



\section{Uhlmann's theorem for relative entropy} \label{sec_main}
In this section, we extend the results from~\cref{sec_main_renyi} to $\alpha =1$. 
\subsection{Regularized Uhlmann's theorem for relative entropy}
\begin{theorem}[Regularized Uhlmann's theorem] \label{thm_uhlmann_rel_ent}
Let $\rho_{AB} \in \St(A \otimes B)$ and $\sigma_A \in \St(A)$. Then
\begin{align} \label{eq_thm_main}
D(\rho_{A}\|\sigma_{A})=   D^{\infty}(\rho_{AB}\| \cC^{\sigma_A}_{AB}) \, .
\end{align}
\end{theorem}
\begin{proof}
We have
\begin{align}
D(\rho_A \| \sigma_A) 
\overset{\textnormal{\Cshref{lem_limits}}}&{=} \lim_{\alpha \searrow 1} D_{\alpha}(\rho_{A}\|\sigma_{A})
\overset{\textnormal{\Cshref{thm_tony_new_range}}}{=} \lim_{\alpha \searrow 1} D_{\alpha}^{\infty}(\rho_{AB}\| \cC^{\sigma_A}_{AB}) 
\overset{\textnormal{\Cshref{lem_limits}}}{=} D^{\infty}(\rho_{AB}\| \cC^{\sigma_A}_{AB})  \, .
\end{align}
\end{proof}

\begin{remark}\label{rmk_explicit_optimizer_regul}
Let $\rho_{AB} \in \St(A \otimes B)$ and $\sigma_A \in \St(A)$. Then, the family\footnotemark
\begin{align} \label{eq_explicit_optimizer_for_n}
\sigma_{A_1^n B_1^n}:=\int_{-\infty}^{\infty} \beta_0(t) (\sigma_{A}^{\frac{1+\ci t}{2}})^{\otimes n} (\rho_{A}^{-\frac{1+\ci t}{2}})^{\otimes n}   \rho_{AB}^{\otimes n} (\rho_{A}^{-\frac{1-\ci t}{2}})^{\otimes n} (\sigma_{A}^{\frac{1-\ci t}{2}})^{\otimes n}  \di t \, ,
\end{align}
of extensions of $\sigma_{A}^{\otimes n}$ for $n\in\N$, where $\beta_0(t):= \frac{\pi}{2}(\cosh(\pi t) +1)^{-1} $ is a probability distribution on $\R$, satisfies
\begin{align}
D(\rho_{A}\|\sigma_{A})
= \lim_{n \to \infty} \frac{1}{n}D(\rho_{AB}^{\otimes n} \| \sigma_{A_1^n B_1^n}) \, .
\end{align}
\end{remark}
\footnotetext{Note that the inverse is only relevant on the support of $\rho_A^{\otimes n}$.}

\begin{proof}
Let $\sigma_{A_1^n B_1^n} \in \cC^{\sigma_A}_{AB,n}$ denote the states defined in~\cref{eq_explicit_optimizer_for_n}. 
One direction follows by noting that
\begin{align}
D(\rho_{A}\|\sigma_{A})
\overset{\textnormal{\Cshref{eq_remove_meas}}}&{=} \lim_{n \to \infty} \frac{1}{n} \MD(\rho_{A}^{\otimes n}\|\sigma_{A}^{\otimes n}) \\
\overset{\textnormal{DPI}}&{\leq}  \lim_{n \to \infty} \frac{1}{n} \MD(\rho_{AB}^{\otimes n} \| \sigma_{A_1^n B_1^n}) \\
\overset{\textnormal{DPI}}&{\leq}  \lim_{n \to \infty} \frac{1}{n} D(\rho_{AB}^{\otimes n} \| \sigma_{A_1^n B_1^n}) \, .
\end{align}
The other direction is correct because the relative entropy is additive under the tensor product and hence
\begin{align}
D(\rho_{A}\|\sigma_{A})
&=\lim_{n \to \infty} \frac{1}{n} D(\rho_{A}^{\otimes n}\|\sigma_{A}^{\otimes n}) \\
\overset{\textnormal{\Cshref{rmk_explicit_extension}}}&{\geq}  \lim_{n \to \infty} \frac{1}{n} \MD(\rho_{AB}^{\otimes n} \| \sigma_{A_1^n B_1^n})\\
\overset{\textnormal{\Cshref{fact_asymptotic_DM_perminv}}}&{=} \lim_{n \to \infty} \frac{1}{n} D(\rho_{AB}^{\otimes n} \| \sigma_{A_1^n B_1^n}) \, ,
\end{align}
where the final step uses that $ \sigma_{A_1^n B_1^n}$ is permutation-invariant for all $n\in\N$.
\end{proof}


\begin{remark}\label{rmk_large}
\Cref{thm_uhlmann_rel_ent} can be viewed as a novel formula for $D(\rho_{A} \| \sigma_{A})$. More precisely, for any two density matrices $\rho_{A}, \sigma_A \in \St(A)$, we can consider an arbitrary purification $\ket{\psi}_{AB}$ of $\rho_A$. Then we have
\begin{align}
D(\rho_{A}\|\sigma_{A})
\overset{\textnormal{\Cshref{thm_uhlmann_rel_ent}}}&{=} \lim_{n \to \infty} \frac{1}{n} \min_{\sigma_{A_1^n B_1^n}\in \cC^{\sigma_A}_{AB,n}} D(\proj{\psi}^{\otimes n} \| \sigma_{A_1^n B_1^n}) \\
\overset{\textnormal{\Cshref{eq_def_D}}}&{=} - \lim_{n \to \infty} \frac{1}{n} \max_{\sigma_{A_1^n B_1^n}\in\cC^{\sigma_A}_{AB,n}} \bra{\psi}^{\otimes n} \log(\sigma_{A_1^n B_1^n}) \ket{\psi}^{\otimes n} \, .
\end{align}
The optimizer can be assumed to be permutation-invariant in $n$. More precisely, as explained in~\cref{rmk_explicit_optimizer_regul}, the optimizer has the explicit form given by~\cref{eq_explicit_optimizer_for_n}.
\end{remark}

\begin{remark}
\Cref{thm_tony_new_range,thm_uhlmann_rel_ent} together with~\cref{eq_tony_counter_example} show that there exist $\rho_{AB} \in \St(A \otimes B)$ and $\sigma_A \in \St(A)$ such that for all $\alpha \in (\frac{1}{2},\infty)$,
\begin{align} \label{eq_regul_necessary}
D_{\alpha}^{\infty}(\rho_{AB}\| \cC^{\sigma_A}_{AB}) < D_{\alpha}(\rho_{AB}\| \cC^{\sigma_A}_{AB}) \, .
\end{align}
This proves that $D_{\alpha}(\rho_{AB}\| \cC^{\sigma_A}_{AB})$ is not additive under the tensor product. However, note that~\cref{eq_regul_necessary} holds with equality for  $\alpha \in \{\frac{1}{2},\infty \}$ as shown in~\cref{rmk_min_max_simplify}.\footnotemark  
\end{remark}
\footnotetext{For $\alpha \in \{\frac{1}{2},\infty \}$, Fekete's subadditivity lemma~\cite{fekete23} actually implies that
\begin{align}
 \frac{1}{n} D_{\alpha}(\rho_{AB}^{\otimes n}\| \cC^{\sigma_A}_{AB,n}) = D_{\alpha}(\rho_{AB}\| \cC^{\sigma_A}_{AB})  \quad \forall n \in \N \, .
\end{align}}

\subsection{Measured Uhlmann's theorem for relative entropy}

\begin{theorem}[Measured Uhlmann's theorem] \label{thm_key}
Let $\rho_{AB} \in \St(A \otimes B)$ and $\sigma_A \in \St(A)$. Then
\begin{align} \label{eq_lemma_important}
D_{\mathbb{M}}(\rho_A \| \sigma_A) \leq D_{\mathbb{M}}(\rho_{AB}\| \cC^{\sigma_A}_{AB}) \leq D(\rho_A \| \sigma_A) \, .
\end{align}
\end{theorem}
\begin{proof}
The first inequality in~\cref{eq_lemma_important} follows from the DPI for the measured relative entropy~\cite[Proposition~2.35]{Sutter_book}. To see the second inequality, note that
\begin{align}
D_{\mathbb{M}}(\rho_{AB} \| \cC^{\sigma_A}_{AB})
\overset{\textnormal{\Cshref{lem_limits}}}{=} \lim_{\alpha \nearrow 1} D_{\alpha,\mathbb{M}}(\rho_{AB}\| \cC^{\sigma_A}_{AB})
\overset{\textnormal{\Cshref{thm_single_letter_Uhlmann_generalization}}}{\leq} \lim_{\alpha \nearrow 1} D_{\alpha}(\rho_{A}\| \sigma_A)
\overset{\textnormal{\Cshref{lem_limits}}}{=}  D(\rho_{A}\| \sigma_A) \, .
\end{align}
\end{proof}

\begin{remark} \label{rmk_explicit_extension}
An alternative proof for the assertion of~\cref{thm_key} based on the multivariate Golden-Thompson inequality~\cite{SBT16} can be found in~\cref{sec_alternative_pf_thm_key}. The alternative proof gives an explicit construction of an extension
\begin{align} \label{eq_extension_QMC}
\bar \sigma_{AB}:=\int_{-\infty}^{\infty} \beta_0(t) \sigma_{A}^{\frac{1+\ci t}{2}} \rho_{A}^{-\frac{1+\ci t}{2}}   \rho_{AB} \rho_{A}^{-\frac{1-\ci t}{2}} \sigma_{A}^{\frac{1-\ci t}{2}}  \di t 
\end{align}
of $\sigma_A$ for a probability distribution $\beta_0(t):= \frac{\pi}{2}(\cosh(\pi t) +1)^{-1} $ on $\R$ that satisfies
\begin{align}
D_{\mathbb{M}}(\rho_A \| \sigma_A) \leq \MD(\rho_{AB}\|\bar \sigma_{AB}) \leq D(\rho_A \| \sigma_A) \, .
\end{align}
\end{remark}

The extension given in~\cref{eq_extension_QMC} as well as the alternative proof from~\cref{sec_alternative_pf_thm_key} has similarities with the recovery maps for approximate quantum Markov chains~\cite{FR14,BHOS14,SFR15,wilde15,SBT16,Sutter_book}.


\subsection{Additional properties of the relative entropy} \label{sec_justification_rmk}
 Next, we generalize a few lemmas from~\cref{sec_main_renyi} to the case $\alpha=1$. These results may be of independent interest.
\begin{lemma} \label{lem_remove_M}
Let $\rho_{AB} \in \St(A \otimes B)$ and $\sigma_A \in \St(A)$. Then
\begin{align} \label{eq_lemma_remove_M}
D^{\infty}_{\mathbb{M}}(\rho_{AB}\| \cC^{\sigma_A}_{AB}) = D^{\infty}(\rho_{AB}\| \cC^{\sigma_A}_{AB})\, .
\end{align}
\end{lemma}
\begin{proof} 
The DPI implies that for any $n \in \N$
 \begin{align}
  \frac{1}{n} D_{\mathbb{M}}(\rho_{AB}^{\otimes n} \| \cC^{\sigma_A}_{AB,n}) \overset{\textnormal{\Cshref{eq_Ua_vs_Ua_meas}}}{\leq}   \frac{1}{n} D(\rho_{AB}^{\otimes n} \| \cC^{\sigma_A}_{AB,n})  \, .
 \end{align}
 Taking the limit $n \to \infty$ yields
 \begin{align}
 D^{\infty}_{\mathbb{M}}(\rho_{AB}\| \cC^{\sigma_A}_{AB}) \leq D^{\infty}(\rho_{AB}\| \cC^{\sigma_A}_{AB})\, .
 \end{align}
 
 Hence, it remains to prove the other direction.
 The unitary invariance and convexity of the measured relative entropy~\cite[Proposition 2.35]{Sutter_book} implies that for any $\bar \sigma_{A_1^n B_1^n} \in \cC^{\sigma_A}_{AB,n}$ we have
 \begin{align}
 \MD(\rho_{AB}^{\otimes n} \| \bar \sigma_{A_1^n B_1^n}) 
 \overset{\textnormal{unitary invariance}}&{=} \frac{1}{n!} \sum_{\pi \in \cS_n}  \MD\big(\pi \rho_{AB}^{\otimes n} \pi^\dagger \|  \pi \bar \sigma_{A_1^n B_1^n} \pi^\dagger \big) \\
 &=\frac{1}{n!} \sum_{\pi \in \cS_n}  \MD\big(\rho^{\otimes n} \|  \pi \bar \sigma_{A_1^n B_1^n} \pi^\dagger \big) \\
 \overset{\textnormal{convexity}}&{\geq}  \MD \Big( \rho^{\otimes n} \Big\|   \frac{1}{n!} \sum_{\pi \in \cS_n}   \pi \bar \sigma_{A_1^n B_1^n} \pi^\dagger  \Big) \, .  \label{eq_opt_PI} 
 \end{align}
 Furthermore, since 
 \begin{align}
  \tr_{B_1^n}\Big[\frac{1}{n!} \sum_{\pi \in \cS_n}   \pi \bar \sigma_{A_1^n B_1^n} \pi^\dagger\Big]
  \overset{\textnormal{\Cshref{eq_valid_extentision}}}{=}\sigma_A^{\otimes n} \, ,
 \end{align}
 it follows that
 \begin{align} \label{eq_optimizer_Er}
  \hat \sigma_{A_1^n B_1^n} \in \argmin_{\bar \sigma_{A_1^n B_1^n} \in \cC^{\sigma_A}_{AB,n}}  \MD(\rho_{AB}^{\otimes n} \| \bar \sigma_{A_1^n B_1^n}) 
 \end{align}
can be assumed to be permutation-invariant. By~\cref{fact_asymptotic_DM_perminv}, it thus follows that
\begin{align}
    \frac{1}{n} D_{\mathbb{M}}(\rho_{AB}^{\otimes n}\| \cC^{\sigma_A}_{AB,n})
    \overset{\textnormal{\Cshref{eq_optimizer_Er}}}&{=} \frac{1}{n}  \MD(\rho_{AB}^{\otimes n} \| \hat \sigma_{A_1^n B_1^n}) \\
    \overset{\textnormal{\Cshref{fact_asymptotic_DM_perminv}}}&{\geq} \frac{1}{n}  D( \rho_{AB}^{\otimes n} \| \hat \sigma_{A_1^n B_1^n} ) + \frac{o(n)}{n} \\ 
    &\geq \frac{1}{n} D(\rho_{AB}^{\otimes n}\| \cC^{\sigma_A}_{AB,n}) + \frac{o(n)}{n} \, .
\end{align}
 Taking the limit $n \to \infty$ yields
 \begin{align}
 D^{\infty}_{\mathbb{M}}(\rho_{AB} \| \cC^{\sigma_A}_{AB}) \geq D^{\infty}(\rho_{AB}\| \cC^{\sigma_A}_{AB})\, .
 \end{align}
\
\end{proof}

Next, we derive upper and lower bounds for $D^{\infty}(\rho_{AB}\| \cC^{\sigma_A}_{AB})$. 
To do so, we show an asymptotic superadditivity property of $D_{\mathbb{M}}(\rho_{A_1^n B_1^n} \| \cC^{\sigma_A}_{AB,n})$ for permutation-invariant states $\rho_{A_1^n B_1^n}$.
\begin{lemma} \label{lem_superadditivity}
 Let $\omega_{AB} \in \St(A\otimes B)$, $\sigma_A \in \St(A)$, and $\{\rho_{A_1^n B_1^n}\}_{n \in \N}$ be a family of states such that $\rho_{A_1^n B_1^n}\in \St(A_1^n \otimes B_1^n)$ and $\rho_{A_i B_i} = \omega_{AB}$ for all $i\in \{1, \ldots ,n \}$. Then,
\begin{align} \label{eq_prop_superadd}
\lim_{n \to \infty}\frac{1}{n}  D_{\mathbb{M}}(\rho_{A_1^n B_1^n}\| \cC^{\sigma_A}_{AB,n}) \geq  D_{\mathbb{M}}(\omega_{AB}\| \cC^{\sigma_A}_{AB})  \, . 
\end{align}
\end{lemma}
\begin{proof}
We are now equipped with all the ingredients to show~\cref{eq_prop_superadd}. To shorten notation, we write $\rho_{A_1^n B_1^n} = \rho_n$ and observe that
\begin{align}
\frac{1}{n}  D_{\mathbb{M}}(\rho_{A_1^n B_1^n}\| \cC^{\sigma_A}_{AB,n})
&= \frac{1}{n}  \min_{\sigma_n\in \cC^{\sigma_A}_{AB,n}} \MD(\rho_n \| \sigma_n) \\
\overset{\textnormal{\Cshref{eq_varFormula_DM}}}&{=}\frac{1}{n}  \min_{\sigma_n\in \cC^{\sigma_A}_{AB,n}} \sup_{\tau_n > 0} \{ \tr[\rho_n \log \tau_n] - \log \tr[\sigma_n \tau_n] \} \\
&\geq\frac{1}{n}  \sup_{\tau_n > 0} \min_{\sigma_n\in \cC^{\sigma_A}_{AB,n}} \{ \tr[\rho_n \log \tau_n] - \log \tr[\sigma_n \tau_n] \} \\
&=\frac{1}{n}  \sup_{\tau_n > 0} \{ \tr[\rho_n \log \tau_n] - \log h_{\cC^{\sigma_A}_{AB,n}}(\tau_n) \} \\
&\geq \sup_{\tau > 0}\{ \tr[\omega \log \tau]  - \frac{1}{n} \log h_{\cC^{\sigma_A}_{AB,n}}(\tau^{\otimes n})  \}  \\
\overset{\textnormal{\Cshref{lem_additivity_support_function}}}&{=}\sup_{\tau > 0}\{ \tr[\omega \log \tau]  -  \log h_{\cC^{\sigma_A}_{AB}}(\tau) \} + \frac{o(n)}{n} \, , \label{eq_prop_almost_done}
\end{align}
where the penultimate step uses that $\log(X_A \otimes Y_B) =\log(X_A) \otimes \id_B + \id_A \otimes \log(Y_B)$ and that by assumption all marginals of $\rho_{A_1^n B_1^n}$ are equal to $\omega_{AB}$.
Using that $1-x \leq -\log x$ for $x \in \R_+$ yields
\begin{align}
\sup_{\tau > 0}\{ \tr[\omega \log \tau]  -  \log h_{\cC^{\sigma_A}_{AB}}(\tau) \} 
&\geq \sup_{\tau > 0}\{ \tr[\omega \log \tau] +1 -  h_{\cC^{\sigma_A}_{AB}}(\tau) \}\\
&= \sup_{\tau > 0} \min_{\sigma \in \cC^{\sigma_A}_{AB}}\{ \tr[\omega \log \tau] +1 -  \tr[\tau \sigma] \} \\
\overset{\textnormal{Sion's minimax~\cite{Sion58}}}&{=} \min_{\sigma \in \cC^{\sigma_A}_{AB}}\sup_{\tau > 0} \{ \tr[\omega \log \tau] +1 -  \tr[\tau \sigma] \} \\
\overset{\textnormal{\Cshref{eq_varFormula_DM2}}}&{=}  \min_{\sigma_{AB} \in \cC^{\sigma_A}_{AB}}\MD(\omega_{AB} \| \sigma_{AB}) \\
\overset{\textnormal{\Cshref{eq_def_UM}}}&{=}D_{\mathbb{M}}(\omega_{AB}\| \cC^{\sigma_A}_{AB}) \, . \label{eq_prop_almost_done2}
\end{align}
Sion's minimax theorem~\cite[Corollary~3.3]{Sion58} is applicable since the set of positive definite matrices and $\cC^{\sigma_A}_{AB}$ are both convex. In addition, the set $\cC^{\sigma_A}_{AB}$ is compact since we are in a finite-dimensional setting. Furthermore, the function $\tau \mapsto \tr[\omega \log \tau] +1 -  \tr[\tau \sigma]$ is concave and $\sigma \mapsto \tr[\omega \log \tau] +1 -  \tr[\tau \sigma]$ is linear.
Combining~\cref{eq_prop_almost_done,eq_prop_almost_done2} completes the proof.
\end{proof}

\begin{remark}
Let $\alpha \in [\frac{1}{2},\infty]$, $\rho_{AB} \in \St(A \otimes B)$, and $\sigma_A \in \St(A)$ be such that $[\rho_A, \sigma_A]=0$. Then
\begin{align} \label{eq_commutative_case}
D_{\alpha,\mathbb{M}}(\rho_{AB}\| \cC^{\sigma_A}_{AB}) =D^{\infty}_{\alpha,\mathbb{M}}(\rho_{AB}\| \cC^{\sigma_A}_{AB}) \, .
\end{align}
One direction follows by noting that
\begin{align}
D_{\alpha,\mathbb{M}}(\rho_{AB}\| \cC^{\sigma_A}_{AB}) 
\overset{\textnormal{\Cshref{lem_superadd_renyi,lem_superadditivity}}}{\leq}D^{\infty}_{\alpha,\mathbb{M}}(\rho_{AB}\| \cC^{\sigma_A}_{AB}) \, .
\end{align}
To see the other direction, we use that $[\rho_A, \sigma_A]=0$ and hence
\begin{align}
D_{\alpha,\mathbb{M}}(\rho_{AB}\| \cC^{\sigma_A}_{AB}) 
\overset{\textnormal{DPI}}&{\geq} D_{\alpha,\mathbb{M}}(\rho_A \| \sigma_A) \\
\overset{\textnormal{\Cshref{eq_euality_DM_vs_D}},\eqref{eq_euality_DM_vs_D_ONE}}&{=}D_{\alpha}(\rho_A \| \sigma_A) \\
\overset{\textnormal{\Cshref{thm_tony_new_range,thm_uhlmann_rel_ent}}}&{=} D^{\infty}_{\alpha}(\rho_{AB}\| \cC^{\sigma_A}_{AB})\\
\overset{\textnormal{\Cshref{lem_remove_M_renyi,lem_remove_M}}}&{=} D^{\infty}_{\alpha,\mathbb{M}}(\rho_{AB}\| \cC^{\sigma_A}_{AB})\, .
\end{align}
\end{remark}

\begin{remark} \label{rmk_bounds_on_U}
A consequence of~\cref{lem_superadditivity} and~\cref{lem_remove_M} is that we can derive single-letter upper and lower bounds for the regularized quantity $D^{\infty}(\rho_{AB}\| \cC^{\sigma_A}_{AB})$ because
\begin{align}
 D_{\mathbb{M}}(\rho_{AB}\| \cC^{\sigma_A}_{AB}) 
 \overset{\textnormal{\Cshref{lem_superadditivity}}}{\leq} D^{\infty}_{\mathbb{M}}(\rho_{AB}\| \cC^{\sigma_A}_{AB}) 
 \overset{\textnormal{\Cshref{lem_remove_M}}}{=}D^{\infty}(\rho_{AB}\| \cC^{\sigma_A}_{AB})
 \leq  D(\rho_{AB}\| \cC^{\sigma_A}_{AB}) \, .
\end{align}
Note that $ D_{\mathbb{M}}(\rho_{AB}\| \cC^{\sigma_A}_{AB})$ as well as $D(\rho_{AB}\| \cC^{\sigma_A}_{AB})$ are convex optimization problems. The latter can be numerically computed via semidefinite programming solvers~\cite{Fawzi_2018}. The former can be computed using interior point methods as explained in~\cite[Remark~2]{FFF24}.
\end{remark}


\paragraph{Acknowledgments.} We thank Henrik Wilming and Joe Renes for making us aware of the argument explained in~\cref{app_Henrik} and pointing us towards~\cite{jen21}, respectively. We also thank Ernest Tan for reminding us about~\cite[Appendix~H.1]{ernest_phd}.
We further acknowledge discussions with Robert K\"onig and Lukas Schmitt about~\cite{FFF24}. GM and RR acknowledge support from the NCCR SwissMAP, the ETH Zurich Quantum Center, the SNSF project No.\ 20QU-1\_225171, and the CHIST-ERA project MoDIC.
\appendix
\section{No sufficient divergence can have the Uhlmann property} \label{app_Henrik}
Let $\TPCP(A,B)$ denote the set of trace-preserving completely positive maps from $A$ to $B$.
We say that a divergence $\mathbb{D}$ satisfying the DPI is \emph{sufficient} if for all $\cE \in \TPCP(A,B)$ and $\rho, \sigma \in \St(A)$ the following two conditions are equivalent:
\begin{enumerate}[(i)]
\item $\mathbb{D}(\rho \| \sigma) = \mathbb{D}\big( \cE(\rho) \| \cE(\sigma) \big)$
\item $\exists\, \cR \in \TPCP(B,A)$ such that $(\cR \circ \cE)(\rho) = \rho$ and $(\cR \circ \cE)(\sigma) = \sigma$. 
\end{enumerate}
The R\'enyi relative entropy $D_{\alpha}$ is sufficient for $\alpha \in (\frac{1}{2},\infty)$~\cite{Pet86,Petz88,jen17,jen21}.\footnote{Note that the Petz R\'enyi relative entropy~\cite{petz_Entropybook,marco_book} is sufficient for $\alpha \in (0,2)$~\cite{jen17}.}
Recall that the R\'enyi relative entropy is lower semi-continiuous~\cite[Proposition~III.11]{MH24}.
In addition, we say $\mathbb{D}$ satisfies the \emph{Uhlmann property} if~\cref{eq_DPI} holds with equality.
\begin{fact} \label{fact_henrik}
Let $\mathbb{D}$ be a lower semi-continuous divergence that satisfies a DPI. Then $\mathbb{D}$ being sufficient implies it does not satisfy the Uhlmann property. 
\end{fact}
\begin{proof}
We prove the assertion of~\cref{fact_henrik} by contradiction.
Suppose we have lower semi-continuous sufficient divergence $\mathbb{D}$ that satisfies the Uhlmann property. This means that for $\rho_A , \sigma_A \in \St(A)$ there exists an extension $\sigma_{AB} \in \cC^{\sigma_A}_{AB}$ such that\footnote{The lower semi-continuity of $\mathbb{D}$ implies the existence of an optimizer $\sigma_{AB} \in \cC^{\sigma_A}_{AB}$, as explained in~\cref{footnote_minimum}.} 
\begin{align}
\mathbb{D}(\rho_A \| \sigma_A) 
= \mathbb{D}(\rho_{AB} \| \sigma_{AB}) 
= \mathbb{D}(\rho_{AB} \|  \cC^{\sigma_A}_{AB}) \, . 
\end{align}
The sufficiency of $\mathbb{D}$ thus implies the existence of a recovery map $\cR \in \TPCP(A,A\otimes B)$ satisfying
\begin{align} \label{eq_sufficiency_henrik}
\cR(\rho_A) = \rho_{AB} \qquad \textnormal{and} \qquad \cR(\sigma_A) = \sigma_{AB} \, .
\end{align}

For an arbitrary divergence $\mathbb{\bar D}$ satisfying the DPI we thus find
\begin{align}
 \mathbb{\bar D}(\rho_{AB} \| \sigma_{AB}) 
\overset{\textnormal{DPI}}{\geq}  \mathbb{ \bar D}(\rho_A \| \sigma_A) 
\overset{\textnormal{DPI}}{\geq} \mathbb{\bar D}\big(\cR(\rho_A) \| \cR(\sigma_A )\big) 
\overset{\textnormal{\Cshref{eq_sufficiency_henrik}}}{=}\mathbb{\bar D}(\rho_{AB} \| \sigma_{AB}) \, . \label{eq_hen1}
\end{align}
This implies
\begin{align}
\mathbb{\bar D}(\rho_A \| \sigma_A) 
\overset{\textnormal{\Cshref{eq_hen1}}}{=} \mathbb{\bar D}(\rho_{AB} \| \sigma_{AB}) 
\overset{\textnormal{DPI}}{=} \mathbb{\bar D}(\rho_{AB} \|  \cC^{\sigma_A}_{AB}) \, ,
\end{align}
which states that $\mathbb{\bar D}$ satisfies the Uhlmann property. This however is in contradiction with~\cite[Appendix~B]{GEAT_24}, where it was shown that there exist $\rho_{AB}\in \St(A\otimes B)$ and $\sigma_A \in \St(A)$ such that $D_{\alpha}(\rho_A \| \sigma_A) < D_{\alpha}(\rho_{AB} \|  \cC^{\sigma_A}_{AB})$ for all $\alpha \in [1,2]$.
\end{proof}

\section{Various technical results}
\subsection{Properties of the measured R\'enyi relative entropy}
\begin{fact}[{\cite[Equation (3.23)]{RSB24}}]\label{fact_UI}
Let $\alpha \in [\frac{1}{2},1) \cup (1,\infty)$, $\rho, \sigma \in \St(A)$, and $U$ be a unitary matrix on $A$. Then
\begin{align}
D_{\alpha,\mathbb{M}}(U \rho U^\dagger \| U \sigma U^{\dagger}) = D_{\alpha,\mathbb{M}}(\rho \| \sigma) \, .
\end{align}
\end{fact}

\begin{fact} \label{fact_convex_concave}
Let $\rho \in \St(A)$. The function $\sigma \mapsto Q_{\alpha,\mathbb{M}}(\rho \| \sigma)$ is concave for $\alpha \in [\frac{1}{2},1)$ and convex for $\alpha \in (1,\infty)$ on the set of density matrices.
\end{fact}
\begin{proof}
For $\sigma_1, \sigma_2 \in \St(A)$, $t \in [0,1]$, and $\alpha \in [\frac{1}{2},1)$ we have
\begin{align}
&\hspace{-10mm}Q_{\alpha,\mathbb{M}}\big(\rho \| t \sigma_1 + (1-t) \sigma_2 \big) \nonumber \\
\overset{\textnormal{\Cshref{eq_varFormula_DM_a}}}&{=}  \inf_{\tau>0} \big \{ \alpha \tr[ \rho  \tau^{\frac{\alpha-1}{\alpha}}] + (1-\alpha) \tr[(t\sigma_1 + (1-t) \sigma_2) \tau] \big \} \\
&=  \inf_{\tau>0} \big \{ t \big( \alpha \tr[ \rho  \tau^{\frac{\alpha-1}{\alpha}}] + (1-\alpha) \tr[\sigma_1 \tau] \big) + (1-t) \big( \alpha \tr[ \rho  \tau^{\frac{\alpha-1}{\alpha}}] + (1-\alpha) \tr[\sigma_2 \tau] \big) \big \} \\
&\geq  t \inf_{\tau>0} \big \{ \alpha \tr[ \rho  \tau^{\frac{\alpha-1}{\alpha}}] + (1-\alpha) \tr[\sigma_1 \tau] \big \} + (1-t) \inf_{\tau>0} \big \{ \alpha \tr[ \rho  \tau^{\frac{\alpha-1}{\alpha}}] + (1-\alpha) \tr[\sigma_2 \tau] \big \} \\
\overset{\textnormal{\Cshref{eq_varFormula_DM_a}}}&{=} t Q_{\alpha,\mathbb{M}}(\rho \|  \sigma_1) + (1-t) Q_{\alpha,\mathbb{M}}(\rho \|  \sigma_2) \, .
\end{align}
For $\alpha \in (1,\infty)$, the same argument works.
\end{proof}

\begin{fact} \label{fact_asymptotic_DM_perminv}
For any $n\in \N$, let $\alpha \in [\frac{1}{2}, \infty]$ and $\rho_n,\,\hat{\sigma}_n \in \St(A_1^n)$ be permutation-invariant. Then
\begin{align}
\frac{1}{n}  D_{\alpha,\mathbb{M}}(\rho_n \| \hat \sigma_n) \geq \frac{1}{n}  D_{\alpha}(\rho_n \| \hat \sigma_n )+ \frac{o(n)}{n}\, .
\end{align}
\end{fact}
\begin{proof}
Since $\hat{\sigma}_n$ is permutation-invariant, it holds~\cite[Lemma~2.8]{GEAT_24} that
\begin{align} \label{eq_poly_spec_renyi}
|\spec(\hat \sigma_n)  | = \poly(n) \, .
\end{align}
Therefore, we find
\begin{align}
\frac{1}{n}  D_{\alpha,\mathbb{M}}(\rho_n \| \hat \sigma_n)
\overset{\textnormal{DPI},\eqref{it_pinch_CPTP}}&{\geq}  \frac{1}{n}  D_{\alpha,\mathbb{M}}\big( \cP_{\hat\sigma}(\rho_n) \| \cP_{\hat\sigma}( \hat \sigma_n) \big) \\
\overset{\eqref{it_pinch_invariance}}&{=}  \frac{1}{n}  D_{\alpha,\mathbb{M}}\big( \cP_{\hat\sigma}(\rho_n) \|  \hat \sigma_n\big) \\
\overset{\eqref{it_pinch_commutation},\textnormal{\Cshref{eq_euality_DM_vs_D}},\eqref{eq_euality_DM_vs_D_ONE}}&{=}\frac{1}{n}  D_{\alpha}\big( \cP_{\hat\sigma}(\rho_n) \|  \hat \sigma_n \big) \\ 
\overset{\textnormal{\cite[Eq.~4.59]{marco_book}}}&{\geq}  \frac{1}{n}  D_{\alpha}( \rho_n \| \hat \sigma_n ) - \frac{2}{n} \log  |\spec(\hat \sigma_{A_1^n B_1^n})  | \\
\overset{\textnormal{\Cshref{eq_poly_spec_renyi}}}&{=} \frac{1}{n} D_{\alpha}(\rho_n \| \hat \sigma_n ) + \frac{o(n)}{n} \, ,
\end{align}
where the DPI for measured R\'enyi relative entropies has been proven in~\cite[Lemma~5]{RSB24}.
\end{proof}
\subsection{Properties of the commutator}
\begin{fact} \label{fact_commute_3_matrices}
Let $X,Y,Z \geq 0$ such that $[X,Y]=[X,Z]=0$, then $[X,YZ]=0$. 
\end{fact}
\begin{proof}
Using that $[X,Y]=[X,Z]=0$ yields
\begin{align}
[X,YZ] = XYZ-YZX = Y(XZ-ZX) = 0 \, .
\end{align}
\end{proof}

\begin{fact} \label{fact_commute_power}
Let $\alpha,\beta>0$ and $X,Y \geq 0$ such that $[X,Y]=0$, then $[X^{\alpha},Y^{\beta}]=0$. 
\end{fact}
\begin{proof}
Let $(\lambda_i)_i$ and $(\mu_i)_i$ denote the eigenvalues of $X$ and $Y$, respectively. Furthermore consider the diagonal matrices $\Lambda=\diag(\lambda_1,\ldots,\lambda_d)$ and $\Xi=\diag(\mu_1,\ldots,\mu_d)$. The condition $[X,Y]=0$ implies the existence of a unitary $U$ such that $X = U \Lambda U^\dagger$ and $Y=U \Xi U^\dagger$. Thus using that $X^\alpha=U \Lambda^\alpha U^\dagger$ and $Y^\beta=U \Xi^\beta U^\dagger$ yields
\begin{align}
[X^{\alpha},Y^{\beta}]
=X^{\alpha}Y^{\beta} - Y^{\beta} X^{\alpha}
= U \Lambda^{\alpha} U^{\dagger} U \Xi^{\beta} U^{\dagger} - U \Xi^{\beta} U^{\dagger} U  \Lambda^{\alpha} U^{\dagger}
=U (\Lambda^{\alpha}  \Xi^{\beta} - \Xi^{\beta}   \Lambda^{\alpha}) U^{\dagger}
=0 \, .
\end{align}
\end{proof}

\subsection{Continuity results}
\begin{lemma} \label{lem_limits}
Let $\rho_{AB} \in \St(A \otimes B)$ and $\sigma_A \in \St(A)$. Then the following statements hold:
\begin{enumerate}[(i)]
\item $\lim_{\alpha \nearrow 1} D_{\alpha}(\rho_A \| \sigma_A)=\lim_{\alpha \searrow 1} D_{\alpha}(\rho_A \| \sigma_A) = D(\rho_A \| \sigma_A)$. \label{it_1_continuity}
\item $\lim_{\alpha \nearrow 1} D_{\alpha,\mathbb{M}}(\rho_A \| \sigma_A) = D_{\mathbb{M}}(\rho_A \| \sigma_A)$.\label{it_2_continuity}
\item $\lim_{\alpha \nearrow 1} D_{\alpha}(\rho_{AB}\| \cC^{\sigma_A}_{AB}) = D(\rho_{AB}\| \cC^{\sigma_A}_{AB})$. \label{it_3_continuity}
\item $\lim_{\alpha \nearrow 1} D_{\alpha,\mathbb{M}}(\rho_{AB}\| \cC^{\sigma_A}_{AB}) = D_{\mathbb{M}}(\rho_{AB}\| \cC^{\sigma_A}_{AB})$. \label{it_3_new_continuity}
\item $\lim_{\alpha \searrow 1} D^{\infty}_{\alpha}(\rho_{AB}\| \cC^{\sigma_A}_{AB}) = D^{\infty}(\rho_{AB}\| \cC^{\sigma_A}_{AB})$. \label{it_4_continuity}
\end{enumerate}
\end{lemma}
\begin{proof}
The statement~\eqref{it_1_continuity} is proven in~\cite[Theorem~5]{MLDSFT13}. Statements~\eqref{it_2_continuity} and~\eqref{it_3_new_continuity} are shown in~\cite[Lemma~23]{FFF24}. 
Property~\eqref{it_3_continuity} follows by the same proof given in~\cite[Lemma~23]{FFF24}.
It thus remains to prove~\eqref{it_4_continuity}. 
The monotonicity of $\alpha \mapsto D_{\alpha}$~\cite[Theorem~7]{MLDSFT13} implies
\begin{align}
\lim_{\alpha \searrow 1} D_{\alpha}^{\infty}(\rho_{AB}\| \cC^{\sigma_A}_{AB})
\overset{\textnormal{\Cshref{eq_Da_regul}}}&{=}\lim_{\alpha \searrow 1} \inf_{n\in \N} \frac{1}{n} D_{\alpha}(\rho_{AB}^{\otimes n} \| \cC^{\sigma_A}_{AB,n}) 
=\inf_{\alpha > 1} \inf_{n\in \N} \frac{1}{n}D_{\alpha}(\rho_{AB}^{\otimes n}\| \cC^{\sigma_A}_{AB,n}) \, . \label{eq_pf_alt_step1}
\end{align}
Since two infima can always be interchanged we find
\begin{align}
\lim_{\alpha \searrow 1} D_{\alpha}^{\infty}(\rho_{AB}\| \cC^{\sigma_A}_{AB})
\overset{\textnormal{\Cshref{eq_pf_alt_step1}}}&{=}\inf_{n\in \N} \inf_{\alpha > 1}  \frac{1}{n}D_{\alpha}(\rho_{AB}^{\otimes n} \| \cC^{\sigma_A}_{AB,n}) \\
\overset{\textnormal{\Cshref{eq_def_Ua}}}&{=}\inf_{n\in \N} \inf_{\alpha > 1} \min_{\sigma_{A_1^n B_1^n} \in \cC^{\sigma_A}_{AB,n}}  \frac{1}{n}D_{\alpha}(\rho_{AB}^{\otimes n}\|\sigma_{A_1^n B_1^n}) \\
&=\inf_{n\in \N} \min_{\sigma_{A_1^n B_1^n} \in \cC^{\sigma_A}_{AB,n}} \inf_{\alpha > 1}   \frac{1}{n}D_{\alpha}(\rho_{AB}^{\otimes n}\|\sigma_{A_1^n B_1^n})\\
\overset{(\dagger)}&=\inf_{n\in \N} \min_{\sigma_{A_1^n B_1^n} \in \cC^{\sigma_A}_{AB,n}}\frac{1}{n}D(\rho_{AB}^{\otimes n}\|\sigma_{A_1^n B_1^n}) \\
\overset{\textnormal{\Cshref{eq_def_Ua}}}&{=} \inf_{n\in \N} \frac{1}{n}D(\rho_{AB}^{\otimes n} \| \cC^{\sigma_A}_{AB,n}) \\
\overset{\textnormal{\Cshref{eq_Da_regul}}}&{=}D^{\infty}(\rho_{AB}\| \cC^{\sigma_A}_{AB}) \, , \label{eq_alt_pf_done}
\end{align}
where the step ($\dagger$) uses again the monotonicity of $\alpha \mapsto D_{\alpha}$ and~\eqref{it_1_continuity}. 
\end{proof}

\begin{lemma} \label{lem_limits_infty}
Let $\rho_{AB} \in \St(A \otimes B)$ and $\sigma_A \in \St(A)$. Then the following statements hold:
\begin{enumerate}[(i)]
\item $\lim_{\alpha \to \infty} D_{\alpha}(\rho_A \| \sigma_A)= \sup_{\alpha>1} D_{\alpha}(\rho_A \| \sigma_A) = D_{\max}(\rho_A \| \sigma_A)$. \label{it_1_continuity_infty}
\item $\lim_{\alpha \to \infty} D_{\alpha,\mathbb{M}}(\rho_A \| \sigma_A) = D_{\max,\mathbb{M}}(\rho_A \| \sigma_A)=D_{\max}(\rho_A \| \sigma_A)$.\label{it_2_continuity_infty}
\item $\lim_{\alpha \to \infty} D_{\alpha}(\rho_{AB} \| \cC^{\sigma_A}_{AB}) = D_{\max}(\rho_{AB} \| \cC^{\sigma_A}_{AB})$.\label{it_3_continuity_infty}
\item $\lim_{\alpha \to \infty} D_{\alpha,\mathbb{M}}(\rho_{AB} \| \cC^{\sigma_A}_{AB}) = D_{\max,\mathbb{M}}(\rho_{AB} \| \cC^{\sigma_A}_{AB}) = D_{\max}(\rho_{AB} \| \cC^{\sigma_A}_{AB})$.\label{it_4_continuity_infty}
\end{enumerate}
\end{lemma}
\begin{proof}
Statements~\eqref{it_1_continuity_infty} and~\eqref{it_2_continuity_infty} are proven in~\cite[Theorems~5 and~7]{MLDSFT13} and~\cite[Appendix A]{MO15}, respectively.
Property~\eqref{it_3_continuity_infty} follows by noting that $\alpha \mapsto D_{\alpha}(\rho_{AB} \| \cC^{\sigma_A}_{AB})$ is monotonically increasing~\cite[Theorem~7]{MLDSFT13}. Hence,
\begin{align}
\lim_{\alpha \to \infty} D_{\alpha}(\rho_{AB} \| \cC^{\sigma_A}_{AB})
&=\sup_{\alpha>1} \min_{\sigma_{AB} \in \cC^{\sigma_A}_{AB}} D_{\alpha}(\rho_{AB} \| \sigma_{AB} ) \\
\overset{\textnormal{\cite[Corollary A.2]{MH11}}}&{=} \min_{\sigma_{AB} \in \cC^{\sigma_A}_{AB}} \sup_{\alpha>1}  D_{\alpha}(\rho_{AB} \| \sigma_{AB} ) \\
&=\min_{\sigma_{AB} \in \cC^{\sigma_A}_{AB}} \lim_{\alpha\to \infty}  D_{\alpha}(\rho_{AB} \| \sigma_{AB} )\\
&=\min_{\sigma_{AB} \in \cC^{\sigma_A}_{AB}}  D_{\max}(\rho_{AB} \| \sigma_{AB} )\\
&=D_{\max}(\rho_{AB} \| \cC^{\sigma_A}_{AB}) \, .
\end{align}
Statement~\eqref{it_4_continuity_infty} follows by the same argument.
\end{proof}

\section{Alternative proof of~\cref{thm_key}} \label{sec_alternative_pf_thm_key}
In this section we present an alternative proof for the assertion of~\cref{thm_key} which has the benefit of yielding an explicit expression for optimizer in $D_{\mathbb{M}}(\rho_{AB}\| \cC^{\sigma_A}_{AB}) $. 
The key mathematical tool for the proof is the Golden-Thompson inequality~\cite{golden65,thompson65} which states that for any two positive definite matrices $A,B > 0$ we have
\begin{align}
\tr[\ee^{\log A + \log B}] \leq \tr[A B] \, .
\end{align}
This inequality has been extended to arbitrarily many matrices in~\cite{SBT16,Hiai2017}. For example, for the case of four positive definite matrices $A,B,C,D > 0$ the multivariate Golden-Thompson inequality states that
\begin{align} \label{eq_4_GT}
\tr[\ee^{\log A + \log B + \log C + \log D}] \leq  \int_{-\infty}^{\infty} \beta_0(t)  \tr[A B^{\frac{1+\ci t}{2}}C^{\frac{1+\ci t}{2}} D  C^{\frac{1-\ci t}{2}} B^{\frac{1-\ci t}{2}}] \, ,
\end{align}
where 
\begin{align}
\beta_0(t):= \frac{\pi}{2} \big(\cosh(\pi t) +1	 \big)^{-1}
\end{align}
is a probability distribution on $\R$.
The interested reader can find more information on the $n$-matrix Golden-Thompson inequality in~\cite{Sutter_book}.

With the help of the multivariate Golden-Thompson inequality, we can now prove the assertion of~\cref{thm_key}.
For the following argument we assume without loss of generality that $\rho_{AB}$ and $\sigma_A$ have full rank.\footnote{If this is not the case one may replace $\rho_{AB}$ and $\sigma_A$ by \smash{$(1-\eps)\rho_{AB} + \eps \frac{\id_{AB}}{d_A d_B} $} and \smash{$(1-\eps)\sigma_A + \eps\frac{\id_A}{d_A}$}, respectively for $\eps >0$. The claim is then obtained in the limit $\eps \to 0$.}
The first inequality in~\cref{eq_lemma_important} follows from the DPI for the measured relative entropy~\cite[Proposition~2.35]{Sutter_book}. To see the second inequality note that
\begin{align} \label{eq_bar_sigma_old}
\bar \sigma_{AB}:=\int_{-\infty}^{\infty} \beta_0(t) \sigma_{A}^{\frac{1+\ci t}{2}} \rho_{A}^{-\frac{1+\ci t}{2}}   \rho_{AB} \rho_{A}^{-\frac{1-\ci t}{2}} \sigma_{A}^{\frac{1-\ci t}{2}}  \di t \, ,
\end{align}
is an extension of $\sigma_A$.\footnote{The matrix $\bar \sigma_{AB}$ defined above is clearly positive semidefinite. In addition, it satisfies $\tr_B[\bar \sigma_{AB}] = \sigma_A$. }
Hence,
\begin{align}
D_{\mathbb{M}}(\rho_{AB}\| \cC^{\sigma_A}_{AB}) 
\overset{\textnormal{\Cshref{eq_def_UM}}}&{=} \min_{\sigma_{AB}} D_{\mathbb{M}}(\rho_{AB} \| \sigma_{AB}) \\
\overset{\textnormal{\Cshref{eq_bar_sigma_old}}}&{\leq} D_{\mathbb{M}}(\rho_{AB} \|  \bar \sigma_{AB}) \\
\overset{\textnormal{\Cshref{eq_varFormula_DM}}}&{=}  \max_{\tau_{AB} >0}\{ \tr[\rho_{AB} \log \tau_{AB}] - \log \tr[\bar \sigma_{AB} \tau_{AB} ] \} \\
\overset{\textnormal{\Cshref{eq_bar_sigma_old}}}&{=}   \max_{\tau_{AB} >0}\Big \{ \tr[\rho_{AB} \log \tau_{AB}] - \log \int_{-\infty}^{\infty} \beta_0(t)  \tr[ \sigma_{A}^{\frac{1+\ci t}{2}} \rho_{A}^{-\frac{1+\ci t}{2}}   \rho_{AB} \rho_{A}^{-\frac{1-\ci t}{2}} \sigma_{A}^{\frac{1-\ci t}{2}}  \tau_{AB} ] \Big\} \\
\overset{\textnormal{\Cshref{eq_4_GT}}}&{\leq} \max_{\tau_{AB} >0}\{ \tr[\rho_{AB} \log \tau_{AB}] - \log \tr[\ee^{ \log \sigma_{A} -  \log \rho_A + \log \rho_{AB} + \log \tau_{AB}}] \} \\
\overset{\textnormal{\Cshref{eq_varFormula_D}}}&{=} D\big(\rho_{AB} \| \ee^{\log\sigma_{A} -  \log \rho_A + \log \rho_{AB}} \big) \\
\overset{\textnormal{\Cshref{eq_def_D}}}&{=}\tr[\rho_{AB} \log \rho_{AB}] - \tr[\rho_{AB}(\log \sigma_{A}- \log\rho_A+ \log\rho_{AB})] \\
\overset{\textnormal{\Cshref{eq_def_D}}}&{=}D(\rho_A \| \sigma_A) \, . 
\end{align} 
\qedhere
\section{Optimal extension is not pure in general} \label{app_counterexample_pure}
Consider the states $\rho_{AB} = \proj{\phi}_{AB}$ for $\ket{\phi}_{AB}= \sqrt{\frac{1}{4}} \ket{00}_{AB} + \sqrt{\frac{3}{4}} \ket{11}_{AB}$ and $\sigma_{A}=\frac{1}{2} \id_2$. For this setting we have
\begin{align}
\log \Big(\frac{3}{2} \Big) \overset{(\triangle)}{=} D_{\max}(\rho_A \| \sigma_A) < \min_{\ket{\psi}_{AB}:\tr_{B}[\proj{\psi}_{AB}]=\sigma_A} D_{\max}(\proj{\phi}_{AB} \| \proj{\psi}_{AB})\overset{(\star)}{=} \infty \, .
\end{align}
This shows that, unlike for the min-relative entropy in case $\rho_{AB}$ is pure, the best extension of $\sigma_{A}$ cannot be chosen to be pure as well for the max-relative entropy. The step $(\triangle)$ can be verified by direct calculation via the formula
\begin{align}
D_{\max}(\rho_A \| \sigma_A) = \log \norm{\sigma_A^{-\frac{1}{2}} \rho_A \sigma_A^{-\frac{1}{2}}} = \log \Big(\frac{3}{2} \Big)  \, , 
\end{align}
where $\norm{\cdot}$ denotes the operator norm. Step $(\star)$ follows by noting that
\begin{align}
\min_{\substack{\ket{\psi}_{AB}:\\ \tr_{B}[\proj{\psi}_{AB}]=\sigma_A}} \!\!\! D_{\max}(\proj{\phi}_{AB} \| \proj{\psi}_{AB}) 
=\min_{\substack{\ket{\psi}_{AB}:\\ \tr_{B}[\proj{\psi}_{AB}]=\sigma_A}} \!\!\! \inf\{\lambda: \proj{\phi}_{AB} \leq 2^{\lambda} \proj{\psi}_{AB} \}
= \infty \, ,
\end{align}
where the final step uses that the infimum over $\lambda$ can only take a finite value if $\ket{\phi}$ and $\ket{\psi}$ are parallels, which is not the case because $\ket{\psi}$ and $\ket{\phi}$ are purifications of $\sigma_A$ and $\rho_A$, respectively, where $\sigma_A \ne \rho_A$. 
\bibliographystyle{arxiv_no_month}
\bibliography{bibliofile}

\newpage

\end{document}